\pdfoutput=1
\RequirePackage{ifpdf}
\ifpdf 
\documentclass[pdftex]{sigma}
\else
\documentclass{sigma}
\fi

\usepackage{bbm}
\usepackage{stmaryrd}
\usepackage{array}
\usepackage{mathtools}
\usepackage{enumerate}

\input xy
\xyoption{all} 

\newcommand{\ep}{\varepsilon}

\newcommand{\dd}{\operatorname{d}}

\newcommand{\rmd}{\textnormal{d}}

\newcommand{\N}{\textnormal{N}}

\newcommand{\f}{\textnormal{f}}
\newcommand{\Vect}{\mathfrak{X}}

\DeclareMathOperator{\w}{w}

\newcommand{\catname}[1]{\textnormal{\texttt{#1}}}

\font\black=cmbx10 \font\sblack=cmbx7 \font\ssblack=cmbx5 \font\blackital=cmmib10 \skewchar\blackital='177
\font\sblackital=cmmib7 \skewchar\sblackital='177 \font\ssblackital=cmmib5 \skewchar\ssblackital='177
\font\sanss=cmss11 \font\ssanss=cmss8 
\font\sssanss=cmss8 scaled 600 \font\blackboard=msbm10 \font\sblackboard=msbm7 \font\ssblackboard=msbm5
\font\caligr=eusm10 \font\scaligr=eusm7 \font\sscaligr=eusm5  \font\fraktur=eufm10
\font\sfraktur=eufm7 \font\ssfraktur=eufm5 
\font\bsymb=cmsy10 scaled\magstep2
\def\all#1{\setbox0=\hbox{\lower1.5pt\hbox{\bsymb
 \char"38}}\setbox1=\hbox{$_{#1}$} \box0\lower2pt\box1\;}
\def\exi#1{\setbox0=\hbox{\lower1.5pt\hbox{\bsymb \char"39}}
 \setbox1=\hbox{$_{#1}$} \box0\lower2pt\box1\;}

\def\tx#1{{\fam0\relax#1}}

\newfam\bifam
\textfont\bifam=\blackital \scriptfont\bifam=\sblackital \scriptscriptfont\bifam=\ssblackital

\newfam\blfam
\textfont\blfam=\black \scriptfont\blfam=\sblack \scriptscriptfont\blfam=\ssblack

\newfam\bbfam
\textfont\bbfam=\blackboard \scriptfont\bbfam=\sblackboard \scriptscriptfont\bbfam=\ssblackboard

\newfam\ssfam
\textfont\ssfam=\sanss \scriptfont\ssfam=\ssanss \scriptscriptfont\ssfam=\sssanss
\def\sss#1{{\fam\ssfam\relax#1}}

\newfam\clfam
\textfont\clfam=\caligr \scriptfont\clfam=\scaligr \scriptscriptfont\clfam=\sscaligr

\newfam\frfam
\textfont\frfam=\fraktur \scriptfont\frfam=\sfraktur \scriptscriptfont\frfam=\ssfraktur

\def\hpb#1{\setbox0=\hbox{${#1}$}
 \copy0 \kern-\wd0 \kern.2pt \box0}
\def\vpb#1{\setbox0=\hbox{${#1}$}
 \copy0 \kern-\wd0 \raise.08pt \box0}

\def\pmb#1{\setbox0\hbox{${#1}$} \copy0 \kern-\wd0 \kern.2pt \box0}
\def\pmbb#1{\setbox0\hbox{${#1}$} \copy0 \kern-\wd0
 \kern.2pt \copy0 \kern-\wd0 \kern.2pt \box0}
\def\pmbbb#1{\setbox0\hbox{${#1}$} \copy0 \kern-\wd0
 \kern.2pt \copy0 \kern-\wd0 \kern.2pt
 \copy0 \kern-\wd0 \kern.2pt \box0}
\def\pmxb#1{\setbox0\hbox{${#1}$} \copy0 \kern-\wd0
 \kern.2pt \copy0 \kern-\wd0 \kern.2pt
 \copy0 \kern-\wd0 \kern.2pt \copy0 \kern-\wd0 \kern.2pt \box0}
\def\pmxbb#1{\setbox0\hbox{${#1}$} \copy0 \kern-\wd0 \kern.2pt
 \copy0 \kern-\wd0 \kern.2pt
 \copy0 \kern-\wd0 \kern.2pt \copy0 \kern-\wd0 \kern.2pt
 \copy0 \kern-\wd0 \kern.2pt \box0}

\mathchardef\za="710B 
\mathchardef\zb="710C 
\mathchardef\zg="710D 
\mathchardef\zd="710E 
\mathchardef\zve="710F 
\mathchardef\zz="7110 
\mathchardef\zh="7111 
\mathchardef\zvy="7112 
\mathchardef\zi="7113 
\mathchardef\zk="7114 
\mathchardef\zl="7115 
\mathchardef\zm="7116 
\mathchardef\zn="7117 
\mathchardef\zx="7118 
\mathchardef\zp="7119 
\mathchardef\zr="711A 
\mathchardef\zs="711B 
\mathchardef\zt="711C 
\mathchardef\zu="711D 
\mathchardef\zvf="711E 
\mathchardef\zq="711F 
\mathchardef\zc="7120 
\mathchardef\zw="7121 
\mathchardef\ze="7122 
\mathchardef\zy="7123 
\mathchardef\zf="7124 
\mathchardef\zvr="7125 
\mathchardef\zvs="7126 
\mathchardef\zf="7127 
\mathchardef\zG="7000 
\mathchardef\zD="7001 
\mathchardef\zY="7002 
\mathchardef\zL="7003 
\mathchardef\zX="7004 
\mathchardef\zP="7005 
\mathchardef\zS="7006 
\mathchardef\zU="7007 
\mathchardef\zF="7008 
\mathchardef\zW="700A 
\mathchardef\zC="7009 

\newcommand{\be}{\begin{equation}}
\newcommand{\ee}{\end{equation}}
\newcommand{\ra}{\rightarrow}

\newcommand{\bea}{\begin{eqnarray}}
\newcommand{\eea}{\end{eqnarray}}
\newcommand{\R}{{\mathbb R}}
\newcommand{\Zet}{{\mathbb Z}}

\newcommand{\s}{{\textstyle *}}

\newcommand{\pa}{\partial}
\newcommand{\ti}{\times}

\def\sD{{\sss D}}

\def\sT{{\sss T}}

\def\sV{{\sss V}}

\def\sE{{\sss E}}
\def\sF{{\sss F}}

\def\sh{{\sss h}}

\def\sv{{\sss v}}

\def\xi{\tx{i}}

\def\dt{\xd_{\sss T}}

\def\cM{\cal M}

\def\cF{{\cal F}}

\def\xd{\operatorname{d}}

\def\dt{\xd_{\sT}}

\def\s*{{\scriptstyle *}}

\def\cM{\mathcal{M}}

\newcommand{\XX}{\mathfrak{X}} 

\newcommand{\veps}{\varepsilon}

\newcommand{\id}{\operatorname{id}}

\def\Linr{{\mathbf{L}}} 
\def\pLinr{{\mathbf{l}}} 
\newcommand{\und}{\underline}

\newcommand{\SymmVB}{\catname{SymmVB}}
\newcommand{\GrB}{\catname{GrB}}
\newcommand{\SpVB}{\catname{SpVB}}
\newcommand{\VB}{\catname{VB}}
\newcommand{\GrL}{\catname{GrL}}

\newcommand{\Y}[1]{\prescript{#1}{}{Y}}
\newcommand{\Z}[1]{\prescript{#1}{}{Z}}

\newcommand{\InnerBracket}{\langle \cdot, \cdot \rangle}
\newcommand{\dblInnerBracket}[1]{\langle\langle #1 \rangle\rangle}
\newcommand{\Sgroup}{\mathbb{S}}

\numberwithin{equation}{section}

\newtheorem{Theorem}{Theorem}[section]
\newtheorem{Corollary}[Theorem]{Corollary}
\newtheorem{Lemma}[Theorem]{Lemma}
\newtheorem{Proposition}[Theorem]{Proposition}
 { \theoremstyle{definition}
\newtheorem{Definition}[Theorem]{Definition}
\newtheorem{Example}[Theorem]{Example}
\newtheorem{Remark}[Theorem]{Remark} }

\begin{document}


\newcommand{\arXivNumber}{1512.02345}

\renewcommand{\PaperNumber}{106}

\FirstPageHeading

\ShortArticleName{Polarisation of Graded Bundles}

\ArticleName{Polarisation of Graded Bundles}

\Author{Andrew James BRUCE~$^\dag$, Janusz GRABOWSKI~$^\dag$ and Miko{\l}aj ROTKIEWICZ~$^\ddag$}

\AuthorNameForHeading{A.J.~Bruce, J.~Grabowski and M.~Rotkiewicz}

\Address{$^\dag$~Institute of Mathematics, Polish Academy of Sciences, Poland}
\EmailD{\href{mailto:andrewjamesbruce@googlemail.com}{andrewjamesbruce@googlemail.com}, \href{mailto:jagrab@impan.pl}{jagrab@impan.pl}}

\Address{$^\ddag$~Faculty of Mathematics, Informatics and Mechanics, University of Warsaw, Poland}
\EmailD{\href{mailto:mrotkiew@mimuw.edu.pl}{mrotkiew@mimuw.edu.pl}}

\ArticleDates{Received December 14, 2015, in f\/inal form October 25, 2016; Published online November 02, 2016}

\Abstract{We construct the \emph{full linearisation functor} which takes a graded bundle of degree $k$ (a particular kind of graded manifold) and produces a $k$-fold vector bundle. We fully characterise the image of the full linearisation functor and show that we obtain a~subcategory of $k$-fold vector bundles consisting of \emph{symmetric $k$-fold vector bundles} equipped with a family of morphisms indexed by the symmetric group~${\mathbb S}_k$. Interestingly, for the degree~2 case this additional structure gives rise to the notion of a~\emph{symplectical double vector bundle}, which is the skew-symmetric analogue of a~\emph{metric double vector bundle}. We also discuss the related case of fully linearising $N$-manifolds, and how one can use the full linearisation functor to ``superise'' a graded bundle.}


\Classification{55R10; 58A32; 58A50}

{\small \tableofcontents}

\section{Introduction and background}\label{sec:Intro}
\subsection{Motivation and summary of results}

The term `graded manifold' appears in the literature in various meanings, including those related just to $\mathbb{Z}_2$-gradation and parity, i.e., to \emph{supermanifolds} (see, e.g.,~\cite{Kostant:1977}). Our general under\-standing is in the spirit of Th.~Voronov \cite{Voronov:2001qf}, who def\/ines \emph{graded manifolds} as supermanifolds equipped with a privileged class of atlases in which coordinates are assigned weights in $\Zet$, which in ge\-neral is independent of the Grassmann parity. Moreover, the coordinate changes are decreed to be polynomial in non-zero weight coordinates and respect the weight. An additional condition is that all the non-zero weight coordinates that are Grassmann even must be `cylindrical'. In precise terms, it means that the associated \emph{weight vector field is $h$-complete} (cf.~\cite{Grabowski2013}).

Within the category of graded manifolds, the most important seem to be the non-negatively graded manifolds. If the Grassmann parity of the coordinates coincides with the weight (mod~2), then we have a $N$-manifold (cf.~\cite{Roytenberg:2002,Severa:2005}). In the purely even setting, Grabowski and Rotkie\-wicz~\cite{JG_MR_higher_vb} def\/ine what they referred to as \emph{graded bundles}, being a particular class of non-negatively graded manifolds. They showed that a graded bundle, understood as a manifold with a non-negative grading on its structure sheaf, is equivalent to a manifold equipped with a smooth action of the multiplicative monoid $(\R,\cdot)$ of reals. Such actions they call \emph{homogeneity structures} (we will be more precise shortly). As this action reduced to $(\R_{>0},\cdot)$ is just generated by the weight vector f\/ield, it means, actually by def\/inition, that the weight vector f\/ield is \emph{$h$-complete}.

In this paper we examine how to canonically pass from a graded bundle of degree~$k$ (cf.~\cite{JG_MR_higher_vb}) to a~$k$-fold vector bundle and back. We construct the \emph{full linearisation functor} as the iteration of the linearisation functor (cf.~\cite{Bruce:2014}), and completely characterise its image as a subcategory of the category of $k$-fold vector bundles. A little more carefully, to every graded bundle of degree~$k$ we can associate a $k$-fold vector bundle equipped with a system of morphisms $\sigma_g$, parametrized by the elements of the symmetric group $\Sgroup_k$, satisfying an additional property; we will call such $k$-fold vector bundles \emph{symmetric}.

Heuristically, one should view the full linearisation as a \emph{polarisation} of the admissible changes of local coordinates. That is, we adjoin new coordinates in a natural way as to linearise the polynomial changes of f\/ibre coordinates. The new coordinates essentially come from the procedure of repeated dif\/ferentiation (i.e., repeated application of the tangent functor). In this way we obtain a new enlarged manifold that has the structure of a $k$-fold vector bundle. The original graded bundle can then be recovered from the `linearised version' as the `diagonal' of \emph{holonomic vectors}. The constructions we present mimic the well-known relationship between higher tangent bundles $\sT^kM$ and iterated tangent bundles. We remark that the deceptively simple notion of polarising homogeneous polynomials has been exploited in algebraic geometry, invariant theory, representation theory and dynamical systems, for example. In retrospect, it is not surprising that polarisation plays an important r\^{o}le in the general theory of graded bundles.

Motivation for considering the full linearisation functor comes from the theory of Lie algebroids that carry an extra compatible grading and the corresponding higher Lagrangian mecha\-nics \cite{Bruce:2014b, Bruce:2014}, the linearisation functor is essential in those works. In addition, we must point out the result of Jotz~Lean \cite{JotzLean:2015} who showed that $N$-manifolds of degree $2$ are categorically equivalent to what she calls \emph{metric double vector bundles}. For the specif\/ic case of degree 2, the linearisation as presented in \cite{Bruce:2014} coincides exactly with the full linearisation presented in this paper. Moreover, for the degree 2 case the additional morphism $\sigma = \sigma_{(12)}$ leads to the construction of a \emph{symplectical double vector bundle} (we will def\/ine this notion carefully in due course). Symplectical double vector bundles are essentially the same as metric double vector bundles, the former being def\/ined in terms of a skew-symmetric pairing while the latter a symmetric pairing. In light of the constructions of Jotz Lean \cite{JotzLean:2015}, we see that the subtle dif\/ference in symmetry of the pairing is really due to the dif\/ference between purely even graded bundles and $N$-manifolds. In particular, for $N$-manifolds of degree two the weight one coordinates are anticommuting, while graded bundles live in the strictly commutative world. Thus we have to interchange `symmetric' and `skew-symmetric' in the right places when switching between graded bundles and $N$-manifolds of degree~2.

\looseness=-1 We also draw the reader's attention to the PhD thesis of del Carpio-Marek \cite{Carpio-Marek:2015}, who (independently) established results equivalent to that that Jotz Lean, but in terms of double vector bundles equipped with an involution. Conceptually this approach is similar to our use of symmetric $k$-fold vector bundles. The work of del~Carpio-Marek starts from some of the results of Bursztyn et al.~\cite{Bursztyn:2015+}. According to del~Carpio-Marek, they show that the category of $N$-manifolds of degree two is equivalent to the category of \emph{involutive double vector bundle sequences}. The notion of a~vector bundle sequence is due to Chen et al.~\cite{Chen:2014}. We further remark Jotz Lean and del Carpio-Marek both concentrate solely on $N$-manifolds of degree~2, where our constructions naturally cover higher degree $N$-manifolds (upon minor modif\/ications). The only other work we are aware of that deals with `linearising' higher degree $N$-manifolds is that of Vishnyakova~\cite{Vishnyakova:2015}: she establishes a categorical equivalence of $N$-manifolds of degree $k$ and $k$-fold vector bundles (in the category of supermanifolds) equipped with a family of $k$ odd vector f\/ields. Vishnyakova uses repeated application of the antitangent functor and substructures thereof to build $k$-fold vector bundles from $N$-manifolds. In particular, the family of odd vector f\/ields are essentially de Rham dif\/ferentials associated with each antitangent bundle. We must remark that we became aware of the works of del~Carpio-Marek and Vishnyakova only towards the end of completing this paper.

Philosophically, graded bundles are a natural generalisation of vector bundles (in the category of smooth manifolds). It is thus natural to wonder if there is some analogue of the parity reversion functor for graded bundles. The obvious dif\/f\/iculty is that the admissible f\/ibre coordinate transformations on a graded bundle are in general non-linear (they are polynomial) and so directly mimicking the parity reversion functor for vector bundles fails.

However, one can use the full linearisation functor, coupled with the standard parity reversion functor, to construct a supermanifold from a graded bundle. Alternatively, the full linearisation functor allows one to canonically associate with an arbitrary graded bundle a $\mathbb{Z}_{2}^{k}$-supermanifold in the sense of Covolo, Grabowski and Poncin~\cite{Covolo:2014a,Covolo:2014b}, and Molotkov~\cite{Molotkov:2010}. From many perspectives, the natural superisation of a $k$-fold vector bundle is a $\mathbb{Z}_{2}^{k}$-supermanifold rather than a~standard supermanifold (see \cite[Proposition~6.1]{Covolo:2014a}). Note that the $\mathbb{Z}_{2}^{k}$-superisation is dif\/ferent to the $\mathbb{Z}_{2}$-superisation of $k$-fold vector bundles as def\/ined by Th.~Voronov~\cite{Voronov:2012}. As the parity reversion functor for vector bundles has turned out to be a very important notion, it is hoped that the $\mathbb{Z}_{2}^{k}$-superisation of a graded bundle will also develop into a useful concept.

A further remark is that the notion of the full linearisation should not be confused with the notion of a splitting of a graded bundle into a direct sum of graded vector bundles. Moreover, the resulting $k$-fold vector bundles or the related $\mathbb{Z}_{2}^{k}$-superisations are not (generally) canonically split. Of course, at each `stage' we have a Batchelor--Gaw\c{e}dzki-like theorem (cf.~\cite{Batchelor:1979,Gawedzki:1977}), but in general the splittings are non-canonical. The existence of Batchelor--Gaw\c{e}dzki-like theorems in various categories has been folklore for quite some time. According to our knowledge, the f\/irst proper proof for $N$-manifolds can be found in~\cite{Bonavolonta:2013}, for graded bundles see \cite{Bruce:2014} and for $\mathbb{Z}_{2}^{k}$-supermanifolds \cite{Covolo:2014b}.

It is desirable to generalise all of the considerations of this paper to more general (purely even) $\mathbb{Z}$-graded manifolds, at least as far as possible. However, manifolds that have both positive and negative gradings are intrinsically harder to understand and work with: basic questions about their topological properties and dif\/ferential calculus remain. $\mathbb{Z}$-graded manifolds in general are not f\/ibre bundles and we lose much of our intuition gained from the study of non-negatively graded structures. Even if we are given an atlas with polynomial change of coordinates, the weight vector f\/ield itself does not in general carry this information. Moreover, there is a lack of illustrative examples to help give insight into the general theory. Graded bundles, in contrast, have much better topological properties and there exists many natural examples. Moreover, if one looses the non-negative grading then there is no possibility of describing the geometry in terms of a homogeneity structure: for example homogeneity structures play a fundamental r\^ole in understanding Lie groupoids in the category of graded bundles \cite{Bruce:2014c}. For these reasons, we will not touch upon the linearisation or polarisation of $\mathbb{Z}$-graded manifolds outside the non-negatively graded case.

We summarise the main results of this paper as follows:
\begin{itemize}\itemsep=0pt
\item We present the full linearisation functor, its inverse (the diagonalisation functor), and show the categorical equivalence between graded bundles of degree k and symmetric $k$-fold vector bundles (cf.~Theorem~\ref{thm:equivalence_k}).
\item The relation between the linearisation of a graded bundle of degree~2 and symplectical double vector bundles is carefully explained (cf.~Theorem \ref{thm:equivalence_symmetric-symplectical}) and canonical Lie algebroid structures on the latter are discovered.
\item The analogous result for $N$-manifolds of degree~2 is found in Proposition~\ref{prop:Jotzlean}, which conceptually simplif\/ies the result of Jotz~Lean \cite[Theorem~3.17]{JotzLean:2015}.
\item The question of the superisation of a graded bundle is answered through the full linearisation functor and then passing canonically to a $\mathbb{Z}_{2}^{k}$-supermanifold (cf.~Theorem~\ref{thm:superisation}).
\end{itemize}

\begin{Remark}\looseness=-1
It is possible to consider non-negatively graded manifolds in terms of consistently def\/ined homogeneous local coordinates that do \emph{not} lead to $h$-complete weight vector f\/ields, i.e., in a broader sense than in~\cite{Voronov:2001qf}. In our opinion, the term ``graded manifolds'' should be reserved for this general concept as illustrated by the following examples. Consider $\R^2\supset F = (0, +\infty)\times (0, +\infty)$ and let $x, y\colon  F\rightarrow \R$ be the standard projections on $(0, +\infty)$. Say, we want to assign weights~$1$,~$1$ to the coordinate functions~$x$,~$y$ on~$F$, so the weight vector f\/ield $\Delta \in \mathfrak{X}(F)$ is $\Delta= x \pa_x + y \pa_y$ which is not $h$-complete ($\Delta$ integrates to $(\R_+, \cdot)$-action which can not be extended to the action of the whole monoid $(\R, \cdot)$). Then $(x, y)\mapsto (x' = x^2/y, y'=y)$ is a~dif\/feomorphism $F\ra F$ and $x'$, $y'$ are homogeneous weight $1$ functions with respect to~$\Delta$. Thus we are forced to accept~$(x', y')$ as a~graded coordinate system on $F$ equally `good' as $(x, y)$. However, the manifold $F$ equipped with the class of graded coordinates represented by $(x, y)$, both of weight~1, and the associated weight vector f\/ield~$\Delta$, inherits now almost no properties of a~vector space. For example, we see that af\/f\/ine combinations can not be def\/ined intrinsically on~$F$.

Another toy example of a non-negatively graded manifold modelled on the same mani\-fold~$F$, is constructed by assigning weights $1$, $2$ to the coordinates $x$, $y$, respectively. Then $(x'=x^{1/2} y^{1/4}, y'=y)$ is another coordinate system in the same class as $(x, y)$. Note that we cannot kill intrinsically coordinates of highest weight, i.e., the coordinate $y$ in the example, so~$F$ does not give a tower of f\/ibrations~\eqref{eqn:fibrations} (see below).

In particular, it will be clear that the full linearisation functor, the principal construction of the paper, can \emph{not} be applied to such graded manifolds. From our perspective, such graded manifolds exhibit pathological behaviour and so we will not consider them in the remainder of this paper.

On the other hand, all our constructions will remain valid if we assume that the transformation functions in homogeneous coordinates are polynomial, so for purely even \emph{non-negatively graded manifolds} in the sense of~\cite{Voronov:2001qf}. The assumption about polynomiality is automatically sa\-tis\-f\/ied for graded bundles, so for simplicity we will work only with this class of graded manifolds.
\end{Remark}

\textbf{Arrangement.} In the remainder of this section we recall the basic theory of graded bundles, $n$-fold graded bundles and the linearisation functor. In Section~\ref{sec:GradedBundlesandHigherVectorBundles} we describe the full linearisation functor, its characterisation and the categorical equivalences it establishes. In Section~\ref{sec:Nmanifolds} the methods of the previous sections are slighted modif\/ied to cope with $N$-manifolds. The question of superisation of graded bundles via the full linearisation is addressed in Section~\ref{sec:superisation}.

\subsection[Graded bundles and $n$-fold graded bundles]{Graded bundles and $\boldsymbol{n}$-fold graded bundles}

An important class of graded manifolds are those that carry non-negative grading. For the moment we will consider only purely even manifolds explicitly, although the statements in this subsection generalise to the supercase. We will furthermore require that this grading is associated with a smooth action $h\colon \R\ti F\to F$ of the monoid $(\R,\cdot)$ of multiplicative reals on a mani\-fold~$F$; a \emph{homogeneity structure} in the terminology of~\cite{JG_MR_higher_vb}. This action reduced to $\R_{>0}$ is the one-parameter group of dif\/feomorphism integrating the \emph{weight vector field}, thus the weight vector f\/ield is in this case \emph{$h$-complete}~\cite{Grabowski2013} and only \emph{non-negative integer weights} are allowed. Thus the algebra $\mathcal{A}(F)\subset C^\infty(F)$ spanned by homogeneous functions is $\mathcal{A}(F) = \bigoplus_{i \in \mathbb{N}}\mathcal{A}^{i}(F)$, where $\mathcal{A}^{i}(F)$ consists of homogeneous function of degree~$i$.

Importantly, we have that for $t \neq 0$ the action $h_{t}$ is a dif\/feomorphism of~$F$ and, when $t=0$, it is a smooth surjection $\tau=h_0$ onto $F_{0}=M$, with the f\/ibres being dif\/feomorphic to~$\mathbb{R}^{N}$ (cf.~\cite{JG_MR_higher_vb}). Thus, the objects obtained are particular kinds of \emph{polynomial bundles} $\tau\colon F\to M$ (e.g.,~\cite{Bertram:2014, Voronov:2010_highr_alg}), i.e., f\/ibrations which locally look like $U\times\R^N$ and the change of coordinates (for a~certain choice of an atlas) are polynomial in $\R^N$. For this reason graded manifolds with non-negative weights \emph{and} $h$-complete weight vector f\/ields~$\zD$ are also known as \emph{graded bundles}~\cite{JG_MR_higher_vb}.

\begin{Example}\looseness=-1 If the weight is constrained to be either zero or one, then the weight vector f\/ield is precisely a vector bundle structure on~$F$ and will be generally referred to as an \emph{Euler vector field}.
\end{Example}
\begin{Example} The principle canonical example of a graded bundle is the higher tangent bundle~$\sT^{k}M$; i.e., the $k$-th jets (at zero) of curves $\gamma\colon \mathbb{R} \rightarrow M$. Given a~smooth function~$f$ on a~mani\-fold~$M$, one can construct functions $f^{(\alpha)}$ on $\sT^k M$, where $0\leq \alpha\leq k$, the so called \emph{$(\alpha)$-lifts} of $f$ (see~\cite{Morimoto_Lifts}). They are def\/ined by
\begin{gather*}
f^{(\alpha)}([\gamma]_k):= \left.\frac{\dd^\alpha}{\dd t^\alpha}\right|_{t=0} f(\gamma(t)).
\end{gather*}
where $[\gamma]_k\in \sT^k M$ is the class of the curve $\gamma\colon \R\rightarrow M$. The functions $f^{(k)}\colon \sT^k M\to \R$ and $f^{(1)}\colon \sT M\to \R$ are called the \emph{$k$-complete lift} and the \emph{tangent lift} of~$f$, respectively. A coordinate system~$(x^a)$ on $M$ gives rise to so called \emph{adapted coordinate systems} $\big(x^{a, (\alpha)}\big)_{0\leq \alpha \leq k}$ on $\sT^k M$ in which $x^{a, (\alpha)}$ is of degree~$\za$, i.e.,
\begin{gather*}\zD=\sum_{a,\za}\za x^{a, (\alpha)}\pa_{x^{a, (\alpha)}}.
\end{gather*}
\end{Example}

On a general graded bundle, one can always pick an atlas of $F$ consisting of charts for which we have homogeneous local coordinates $\big(x^{A},y_{w}^{i}\big)$, where $\w\!\big(x^{A}\big) =0$ and $\w(y_{w}^{i}) = w$ with $1\leq w\leq k$, for some $k \in \mathbb{N}$ known as the \emph{degree} of the graded bundle. Note that, according to this def\/inition, a graded bundle of degree $k$ is automatically a graded bundle of degree $l$ for $l\ge k$. However, there is always a \emph{minimal degree}.

It will be convenient to group all the coordinates with non-zero weight together, as these form a basis of the function algebra of the graded bundle. The index $i$ should be considered as a ``generalised index'' running over all the possible weights. The label $w$ in this respect largely redundant, but it will come in very useful when checking the validity of various expressions. The local changes of coordinates respect the weight and hence are polynomial for non-zero weight coordinates. A little more explicitly, the changes of local coordinates are of the form
\begin{gather*}
 x^{A'} = x^{A}(x), \qquad y^{i'}_{w} = \sum \frac{1}{n!} y^{j_{1}}_{w_{1}} y^{j_{2}}_{w_{2}} \cdots y^{j_{n}}_{w_{n}}T_{j_{n} \cdots j_{2} j_{1}}^{\:\:\:\:\:\:\:\:\:\:\:\: \:\:\:\: i'}(x), \end{gather*}
where $w=w_1+\cdots+w_n$ and we assume the tensors $T_{j_{n} \cdots j_{2} j_{1}}^{\:\:\:\:\:\:\:\:\:\:\:\: \:\:\:\: i'}$ to be symmetric in lower indices. Naturally, changes of coordinates are invertible and so, automatically, $\big(T_{j}^{\:\: i'}(x)\big)$ is an invertible matrix.

Importantly, a graded bundle of degree $k$ admits a sequence of surjections
\begin{gather}\label{eqn:fibrations}
F=F_{k} \stackrel{\tau^{k}_{k-1}}{\longrightarrow} F_{k-1} \stackrel{\tau^{k-1}_{k-2}}{\longrightarrow} \cdots \stackrel{\tau^{3}_2}{\longrightarrow} F_{2} \stackrel{\tau^{2}_1}{\longrightarrow}F_{1} \stackrel{\tau^{1}}{\longrightarrow} F_{0} = M,
\end{gather}
where $F_l$ itself is a graded bundle over $M$ of degree $l$, obtained from the atlas of $F_k$ by removing all coordinates of degree greater than~$l$ (see the next paragraph).

Note that $F_{1} \rightarrow M$ is a linear f\/ibration and the other f\/ibrations $F_{l} \rightarrow F_{l-1}$ are af\/f\/ine f\/ibrations in the sense that the changes of local coordinates for the f\/ibres are linear plus an additional additive terms of appropriate weight. The model f\/ibres here are~$\mathbb{R}^{n}$.

Morphisms between graded bundles necessarily preserve the weight; in other words, morphisms relate the respective homogeneity structures, or equivalently morphisms relate the respective weight vector f\/ields. Evidently, morphisms of graded bundles can be composed as standard maps between smooth manifolds and so we obtain the category of graded bundles. We will denote the category of graded bundles of degree $k$ by $\GrB[k]$.

\begin{Remark}
We will also encounter $N$-manifolds in this paper, which are very similar to graded bundles in many respects, and with hindsight are special examples of graded manifolds as def\/ined by Voronov \cite{Voronov:2001qf}. \v{S}evera \cite{Severa:2005} def\/ines an $N$-manifold as a supermanifold equipped with an action of $(\mathbb{R}, \cdot)$ such that $-1 \in \mathbb{R}$ acts as the parity operator (it f\/lips sign of any Grassmann odd coordinate). Roytenberg \cite{Roytenberg:2002} def\/ines an $N$-manifold in terms of an atlas of charts consisting of homogeneous coordinates for which coordinates of even weight are Grassmann even and coordinates of odd weight are Grassmann odd. It turns out that, under the additional assumption that even coordinates of non-zero degree are `cylindrical', both these notions of $N$-manifolds are equivalent, though the f\/irst (sketch) of a proof we know of is to be found in \cite{Bruce:2014c} and a complete proof can be found in \cite{Jozwikowski:2016}.
\end{Remark}

The notion of a double vector bundle \cite{Pradines:1974} (or an $n$-fold vector bundle \cite{Gracia-Saz:2009,Gracia-Saz:2012,Mackenzie:2005,Voronov:2012}) is conceptually clear in the graded language in terms of mutually commuting weight vector f\/ields; see \cite{Grabowski:2006,JG_MR_higher_vb}. This leads to the higher analogues known as \emph{$n$-fold graded bundles}, which are manifolds for which the structure sheaf carries an $\mathbb{N}^{n}$-grading such that all the weight vector f\/ields are $h$-complete and pairwise commuting. The local triviality of $n$-fold graded bundles means that we can always equip an $n$-fold graded bundle with an atlas such that the charts consist of coordinates that are simultaneously homogeneous with respect to the weights associated with each weight vector f\/ield (cf.~\cite{JG_MR_higher_vb}). If all the weight vector f\/ields are in fact Euler vector f\/ields, then we have an \emph{$n$-fold vector bundle}. The changes of local coordinates on an $n$-fold graded bundle must respect the weights. Similarly, morphisms between $n$-fold graded bundles respect the weights of the local coordinates.

We will denote a $k$-fold vector bundle as $\sD = (D; \Delta^1, \ldots, \Delta^k)$, where $\Delta^i$'s are Euler vector f\/ields on an underlying manifold $D$. We will often use the short hand `DVB' for double vector bundle. We shall use a similar notation for a $k$-fold graded bundle, but shall use the letter $\sF = (F; \Delta^1, \ldots, \Delta^k)$ to indicate that weights greater than one are allowed. We remark that the ordering of the weight vector f\/ields is a part of the def\/inition of a $k$-fold graded bundle and that morphisms of $k$-fold vector bundles respect this ordering. Thus, any permutation of the ordering of the weight vector f\/ields gives a dif\/ferent $k$-fold vector bundle structure on the same underlying manifold. The $k$-fold graded bundle $\sF$ gives rise to a number of graded bundles $\tau_i\colon F\rightarrow F_i$ def\/ined by pairs $(F; \Delta^i)$ called \emph{legs} of $\sF$, and a number of $(k-1)$-fold graded bundles $\sF_i = (F_i; \Delta^1, \ldots, \Delta^{i-1}, \Delta^{i+1}, \ldots, \Delta^k)$ which we call \emph{feet} of~$\sF$.

Important examples are: the \emph{iterated tangent bundle} $\sT^{(k)} M:=\sT\cdots\sT M$ (a $k$-fold vector bundle) and~$\sT^k\sT^l M$ (a double graded bundle).

By iterating the construction of $(\za)$-lifts, we obtain functions $f^{(\beta, \alpha)} := (f^{(\beta)})^{(\alpha)}$ on $\sT^k \sT^l M$ for $0\leq \alpha\leq k$, $0\leq \beta\leq l$, and, generally, functions $f^{(\ep_r, \ldots, \ep_1)}$ on $\sT^{n_1} \cdots \sT^{n_r} M$ for $0\leq \ep_j\leq n_j$, $1\leq j\leq r$. A coordinate system $(x^a)$ on $M$ gives rise to so called \emph{adapted coordinate sys\-tems}~$(x^{a, (\ep)})_\ep$ on $\sT^{(k)} M$, where the multi-index $\ep=(\ep_r, \ldots, \ep_1)$ is from $\{0,1\}^k$ and $x^{a, (\alpha)}$, $x^{a, (\ep)}$ are obtained from $x^a$ by the above lifting procedure.

\begin{Definition}\label{def:symmetricVB}
Consider a double vector bundle, which following classical notation we denote as $\sD=(D; A, B; M)$, where $A$ and $B$ are the side bundles. If we are given a vector bundle isomorphism $A \simeq B$ then $\sD$ is called \emph{balanced}. We may simply write $\sD=(D; A, A; M)$ but usually will avoid this. In more generality, a $k$-fold vector bundle~$\sD$ with legs $D\rightarrow D_i$, $i=1, \ldots, k$, is said to be a \emph{balanced $k$-fold vector bundle} if all feet $D_i$ are pairwise-isomorphic as $(k-1)$-fold vector bundles in some canonical way.
\end{Definition}

\subsection{Some constructions}
By saying that a vector f\/ield $\Delta$ on a manifold $F$ \emph{is a weight vector field} we mean that $F$ can be given a graded bundle structure which def\/ines $\Delta$ as its weight vector f\/ield. In particular, $\Delta$~has to be complete and the linear part of $\Delta$ in any of singular point of $\Delta$ has to have integer, non-negative eigenvalues. A important remark \cite[Remark~5.2]{JG_MR_higher_vb} is that the sum (unlike the dif\/ference) of two commuting weight vector f\/ields is again a weight vector f\/ield, hence a linear combination of weight vector f\/ields with non-negative integer coef\/f\/icients $a_s$, $\Delta = \sum_s a_s \Delta^s$ is a~weight vector f\/ield. We shall use this fact to set some further useful notation.

\begin{Definition}[removal coordinates of weight $> l$]\label{def:removal_coordinates1} Let $\sF = (F; \Delta)$ be a graded bundle.
Denote by $F[\Delta \leq l]$ the base manifold of the locally trivial f\/ibration def\/ined by taking the weight $>l$ coordinates with respect to this complementary weight to be the f\/ibre coordinates. We have a~natural projection that we will denote as
\begin{gather*}
\textnormal{p}^F_{[\Delta \leq l]}\colon \  F \rightarrow F[\Delta \leq l].
\end{gather*}
\end{Definition}
Note that, as the transformation rules of coordinates of weights $\le l$ involve only coordinates of weights $\le l$, the above def\/inition is correct. Since compositions of commuting homogeneity structures are homogeneity structures, we immediately get the following.
\begin{Proposition} Linear combinations of commuting weight vector fields with non-negative integer coefficients are weight vector fields. If at start $\sF = (F; \Delta^1, \ldots, \Delta^k)$ is a $k$-fold graded bundle and $\Delta=\sum a_s \Delta^s$, where $a_i\in\N$, then $F[\Delta \leq l]$ has the induced $k$-fold graded bundle structure such that $\textnormal{p}^F_{[\Delta \leq l]}$ is a $k$-fold graded bundle morphism.
\end{Proposition}

The next construction is in a sense dual and depends on putting to zero coordinates of negative weights with respect to the vector f\/ield $X=\sum a_s \Delta^s$. Here, we assume that $a_s\in \Zet$ can be negative (see Remark \ref{rem:non_integer_weights}).
\begin{Definition}[putting to zero coordinates of $X$-negative weights]\label{def:removal_coordinates2} Consider a $k$-fold graded bundle $(F; \Delta^1, \ldots, \Delta^k)$ and let $X=\sum_s a_s \Delta^s$ be a linear combination of weight vector f\/ields of~$F$ with integer coef\/f\/icients. Def\/ine locally a submanifold $F[X\geq 0] \subset F$ as
\begin{gather*}
 F[X\geq 0]:= \bigg\{\big(y^i_w\big)\colon y^i_w=0 \  \text{for all} \  w=(w_1, \ldots, w_k)\  \text{such that}\  \sum_s a_s w_s < 0\bigg\}.
\end{gather*}
\end{Definition}

Note f\/irst that $F[X\geq 0]$ is a well-def\/ined submanifold of $F$. Indeed, set $X$-weight $\w_X(f)$ for the weight with respect to the vector f\/ield $X$ of a homogeneous function $f$, thus $\w_X(y^i_w) = \sum_s a_s w_s$ if $w=(w_1, \ldots, w_s)$, and consider another coordinate system $(y^{i'})$. Since coordinate changes preserve $\w_X$ and $y^{i'}$ is a linear combination of monomials of type $y^{i_1}\cdots y^{i_j}$, we see that if $\w_X(y^{i'})$ is negative then at least one of the weight $\w_X(y^{i_1})$, $\w_X(y^{i_2}), \ldots, \w_X(y^{i_j})$ has to be negative, and thus $y^{i'}$ vanishes on $F[X\geq 0]$, hence the latter is invariantly def\/ined. As the weight vector f\/ields $\Delta^1, \ldots, \Delta^k$ preserve the $X$-weight, they are tangent to the submanifold $F[X\geq 0]$ which therefore has the induced $k$-fold graded bundle structure.

Considering $-X$ instead of $X$, we obtain graded subbundles $F[-X\ge 0]$, thus we can put to zero invariantly all coordinates of positive $X$-weight, and, eventually, graded subbundle \begin{gather*}
F[X=0] = F[X\geq 0] \cap F[-X\geq 0].
\end{gather*}
Is is clear that $F[X=0]$ consist of singular points of the vector f\/ield $X = \sum a_s \Delta^s$.

\begin{Remark}\label{rem:non_integer_weights} Formally we could admit even irrational coef\/f\/icient in $X$ to def\/ine $F[X\geq 0]$ or $F[X=0]$ but this does not enlarge possible constructions, e.g., if $\Delta^1 = x \pa_x$, and $\Delta^2 = y\pa_y$ are weight vector f\/ields of a double graded bundle $F$ then $F[x\pa_x + \sqrt{2} y \pa_y\geq 0]$ is the same as $F[x\pa_x + q y \pa_y\geq 0]$ for some rational number $q$ close to $\sqrt{2}$ since there is only a f\/inite number of weights assigned to coordinates; moreover, $F[X\geq 0] = F[k X\geq 0]$, $k\neq 0$, so we can forget about non-integer coef\/f\/icients in $X$.
\end{Remark}

\begin{Example}
If $(F, \Delta)$ is a graded bundle, then $F[\Delta=0]$ is the base of $F$.
\end{Example}
\begin{Example} \label{exmpl:core} If $\sD = (D; \Delta^1, \Delta^2)= (D; A, B; M)$ is a double vector bundle, then $C:=\sD[\Delta^1-\Delta^2=0]$ is the core bundle of $\sD$. Besides, $\sD[\Delta^1-\Delta^2\geq 0]$ is the kernel of the vector bundle morphism $D\rightarrow B$, so the intersection $\sD[\Delta^1-\Delta^2\geq 0]\cap \sD[\Delta^2-\Delta^1\geq 0]$ coincides with the core bundle.
\end{Example}

\begin{Remark} The projection $\textnormal{p}^F_{[\Delta \leq l]}$ from Def\/inition \ref{def:removal_coordinates1} and the inclusions
$F[X=0] \rightarrow F$ and $F[X\geq 0] \rightarrow F$ are functorial. We shall use this fact later on.
\end{Remark}

\subsection{The na\"{\i}ve approach to superisation}
Let us brief\/ly discuss the problem of the `superisation' of a graded bundle via the explicit example of the degree 2 case: this will be suf\/f\/icient to outline the problem. Consider $F_{2}$ equipped with local coordinates $\big(x^{A},~y^{a},~z^{i}\big)$ of weight $0$, $1$ and $2$, respectively. The admissible changes of local coordinates are of the form
\begin{gather*}
x^{A'} = x^{A'}(x), \qquad y^{a'} = y^{b}T_{b}^{\:\:\: a'}(x), \qquad z^{i'} = z^{j}T_{j}^{\:\:\: i'}(x) + \frac{1}{2}y^{a}y^{b}T_{ba}^{\:\:\:\: i'}(x),
 \end{gather*}
where we take $T_{ba}^{\:\:\:\: a'}$ to be symmetric in the lower indices.

Any meaningful notion of `superisation' should be in terms of an invertible functor and thus establish a categorical equivalence between the category of graded bundles and some subcategory of the category of supermanifolds. Let us by brute force `superise'~$F_2$ by declaring the Grassmann parity of the coordinates to be equal to the weight. However, notice that this procedure does not work properly as the changes of coordinates on~``$\Pi F_{2}$'' would not contain the tensor $T_{ba}^{\:\:\:\: i'}$ due to symmetry of the tensor in the lower indices. Therefore, some information about $F_2$ is lost and so we are unable to recover~$F_2$  from such a `superisation'. This example illustrates a `no-go theorem'.

\emph{There is no canonical nontrivial direct procedure for constructing a one-to-one correspondence between graded bundles and $N$-manifolds.}

 That is, there is no direct analogue of the parity reversion function $\Pi$ for graded bundles. One cannot in general simply declare (maybe up to signs) some coordinates to be Grassmann odd. However, that is not to say that one can never superise graded bundles. An obvious exception here are $k$-fold graded bundles for which one or more of the graded structures is linear. In such cases one can apply the standard parity reversion functor; a great example here are $k$-fold vector bundles or graded-linear bundles.

\subsection{The linearisation of a graded bundle}
 In this subsection we shall recall the construction of the \emph{linearisation functor} given in~\cite{Bruce:2014}. In order to avoid some notational clashes, we have changed some of the notation as compared with previous works.

The tangent bundle $\sT F_{k}$ of a graded bundle naturally has the structure of a double graded bundle. The f\/irst weight vector f\/ield is simply the \emph{tangent lift} \cite{Grabowski:1995,Yano:1973} $\Delta^{1} := \rmd_{\sT}\Delta_{F_{k}}$ of the weight vector f\/ield $\Delta_{F_{k}}$ on $F_{k}$, and the second being the natural Euler vector f\/ield on the tangent bundle $\Delta^{2} := \Delta_{\sT F_{k}}$. The tangent bundle of a~graded bundle is a canonical example of a graded-linear bundle, i.e., a graded bundle equipped with a compatible linear structure on the total space.

\begin{Definition}[\cite{Bruce:2014}] A double graded bundle $\sF = (F; \Delta^{1}, \Delta^{2})$ such that~$\zD^1$ is of degree $k-1$ and $\zD^2$ is of degree~1, i.e., $\Delta^{2}$ is an Euler vector f\/ield, will be referred to as a \emph{graded-linear bundle}, or for short a $\GrL$-bundle.
\end{Definition}

It follows that $\sF$ is of total weight $\le k$ with respect to $\Delta := \Delta^{1} + \Delta^{2}$ and a vector bundle structure
\begin{gather*}
\textnormal{p}^{F}_{[\Delta^{2} \leq 0]}\colon \  F \rightarrow F\big[\Delta^{2} \leq 0\big].
\end{gather*}
with respect to the projection onto the submanifold $F[\Delta^{2} \leq 0]$ which inherits a graded bundle structure of degree $k-1$.

\begin{Example} Consider the vertical bundle with respect to the projection $\tau\colon F_{k} \rightarrow M$. The weight vector f\/ields on the vertical bundle $\sV F_{k} \subset \sT F_{k}$ are simply the appropriate restrictions of those on the tangent bundle. \emph{Via} passing to the total weight, we can view~$\sV F_{k}$ as a graded bundle of degree $k+1$. However, it will be useful to shift the f\/irst component of bi-weight to allow us to consider the vertical bundle as a graded bundle of degree $k$. We will always consider this shifted weight when encountering a vertical bundle of a~graded bundle. In other words, if we take $\Delta^1\in \XX(\sT F_k)$ the tangent lift of the weight vector f\/ield and $\Delta^2\in \XX(\sT F_k)$ the Euler vector f\/ield associated with the tangent bundle structure on $F_k$ we get
\begin{gather*}
\sV F_k = \sT F_k\big[\Delta^1-\Delta^2 \geq 0\big],
\end{gather*}
equipped with the weight vector f\/ield $\Delta^1_{\sV F_k}:= (\Delta^1-\Delta^2)|_{\sV F_k}$ of degree $k$ and the Euler vector f\/ield $\Delta^2|_{\sV F_k}$. It is crucial here that, although it is not a combination with non-negative coef\/f\/icients, $\Delta^1-\Delta^2$ is still a weight vector f\/ield on $\sV F_k$.
\end{Example}

\begin{Example}
 Consider a degree two graded bundle $F_{2}$ which we equip with homogeneous local coordinates $(x,y,z)$ of weight $0$, $1$ and $2$, respectively. Taking the tangent functor combined with the necessary shift in the weight, produces homogeneous local coordinates
 \begin{gather*}
 \Big(\underbrace{x}_{(0,0)}, \hspace{5pt} \underbrace{y}_{(1,0)}, \hspace{5pt} \underbrace{z}_{(2,0)}; \hspace{5pt} \underbrace{\dot{x}}_{(-1,1)},\hspace{5pt}\underbrace{\dot{y}}_{(0,1)}, \hspace{5pt} \underbrace{\dot{z}}_{(1,1)}\Big).
 \end{gather*}
The removal of the coordinate $\dot{x}$ is well def\/ined (see Def\/inition~\ref{def:removal_coordinates2}). Doing so produces a double graded bundle which we recognize as the standard vertical bundle of $F_{2}$.
 \end{Example}

The \emph{linearisation functor} takes a graded bundle and produces a $\GrL$-bundle. The basic idea is to mimic the canonical embedding
\begin{gather*}
\sT^{k}M \hookrightarrow \sT\big(\sT^{k-1}M\big),
 \end{gather*}
which sends a point in $\sT^k M$ represented by a curve $t \mapsto \gamma(t)$ to a vector tangent to $\sT^{k-1} M$ represented by the curve $s\mapsto [t\mapsto \gamma(s+t)]_{k-1}\in \sT^{k-1} M$. The linearisation of a~graded bundle~$F_{k}$ is then viewed as a reduction of the tangent bundle~$\sT F_{k}$.
\begin{Definition}[\cite{Bruce:2014}]\label{def:linearisation}
The \emph{linearisation of a graded bundle} $F_{k}$ is the $\GrL$-bundle def\/ined as
\begin{gather*}
\pLinr (F_{k}) := \sV F_{k}\big[\Delta_{\sV F_{k}}^{1} \leq k-1\big],
\end{gather*}
so that we have the natural projection $\textnormal{p}^{\sV F_{k}}_{\pLinr(F_{k})}\colon \sV F_{k} \rightarrow \pLinr(F_{k})$. With a small abuse of notation, we would write $\pLinr (F_k) = \sT F_k[0\leq \Delta^1 - \Delta^2 \leq k-1]$, i.e., we get rid of coordinates on~$\sT F_k$ of negative weights and of weight~$k$ with respect to the vector f\/ield $\Delta^1 - \Delta^2$.
\end{Definition}

\begin{Theorem}[\cite{Bruce:2014}]\label{theom:functorial linearisation}
Linearisation is a functor from the category of graded bundles to the category of double graded bundles.
\end{Theorem}

A full proof of the above theorem can be found in \cite{Bruce:2014}. Essentially, the functorial properties of the linearisation follow from the fact that it is constructed from the tangent functor.

Let us brief\/ly describe the local structure of the linearisation. Consider $F_{k}$ equipped with local coordinates $\big(x^{A}, y_{w}^{a}, z^{i}_{k}\big)$, where the weights are assigned as $\w\big(x^A\big)=0$, $\w(y^a_{w}) = w$ $(1\leq w < k)$ and $\w(z^i_k) =k$. It will be convenient to single out the highest weight coordinates, as well as the zero weight coordinates, and so we change our notation slightly. In any natural homogeneous system of coordinates on the vertical bundle $\sV F_{k}$, one projects out the highest weight coordinates on~$F_{k}$ to obtain~$\pLinr(F_{k})$. Thus, on $\pLinr(F_{k})$ we have local homogeneous coordinates
\begin{gather*}
\Big(\underbrace{x^{A}}_{(0,0)}, \hspace{5pt} \underbrace{y^{a}_{w}}_{(w,0)}; \hspace{5pt} \underbrace{\dot{y}^{b}_{w}}_{(w-1,1)}, \hspace{5pt} \underbrace{\dot{z}_{k}^{i}}_{(k-1,1)}\Big).
 \end{gather*}
Note that with this assignment of the weights, the linearisation of a graded bundle of degree~$k$ is itself a graded bundle of degree~$k$ when passing from the bi-weight to the total weight. It is important to note that the linearisation of a graded bundle has the structure of a vector bundle $\pLinr(F_{k}) \rightarrow F_{k-1}$, hence the nomenclature ``linearisation''. The vector bundle structure is clear from the construction by examining the bi-weight. The second leg of $\GrL$-bundle $\pLinr(F_{k})$ is a degree $k-1$ graded bundle with base~$F_1$:
\begin{gather*}
\xymatrix{
\pLinr(F_{k}) \ar[rr]\ar[d] && F_1 \ar[d] \\
F_{k-1} \ar[rr] && M.
}
\end{gather*}
The basic properties of the linearisation functor are summarized in the following.
\begin{Proposition}\label{prop:pLin}\quad
\begin{enumerate}[$(a)$]\itemsep=0pt
\item \label{prop_pL:item_pullback} There exist a canonical vector bundle morphism $p=\textnormal{p}^{\sV F_{k}}_{\pLinr(F_{k})} \colon \sV F_{k} \rightarrow \pLinr(F_{k})$ covering the projection $\tau^k_{k-1}\colon F_k\rightarrow F_{k-1}$ such that $\sV F_{k}$ is the pullback of the vector bundle $\pLinr(F_{k})\rightarrow F_{k-1}$ by $p$:
 \begin{gather*}
\xymatrix{
p^* \pLinr(F_{k}) \simeq \sV F_{k}\ar[rr]^p\ar[d] && \pLinr(F_{k}) \ar[d] \\
F_k \ar[rr]^{\tau^k_{k-1}} && F_{k-1}.
}
\end{gather*}
\item \label{prop_pL:itemTkM} The linearisation $\pLinr(\sT^k M)$ of the $k$-th order tangent bundle of $M$ is isomorphic with $\sT \sT^{k-1} M$ as $\GrL$-bundle. In particular, there is a canonical epimorphism $ p\colon \sV \sT^k M = \ker (\sT \tau^k_M\colon \sT\sT^k M $ $\rightarrow \sT M) \twoheadrightarrow \sT\sT^{k-1} M$. In coordinates,
 \begin{gather*}
 p^*\big(x^{a, (\alpha)}\big) = x^{a, (\alpha)}, \qquad p^*\big(dx^{a, (\beta-1)}\big) = dx^{a, (\beta)},
 \end{gather*}
 where $0\leq \alpha\leq k-1$, $1\leq \beta \leq k$.
\item \label{prop_pL:item embedding} There is a canonical $($total$)$ weight preserving embedding
\begin{gather*}
\iota \colon \ F_{k} \hookrightarrow \pLinr(F_{k})
\end{gather*}
given by the composition of the weight vector field $\Delta_{F_k} \in \Vect(F_k)$ considered as a section of~$\sV F_{k}$ with the natural projection $p\colon \sV F_{k} \rightarrow \pLinr(F_{k})$. In natural local coordinates, the nontrivial part of the embedding is given by
\begin{gather*}
\iota^{*}(\dot{y}^{a}_{w}) = w   y^{a}_{w}, \qquad \iota^{*}(\dot{z}_{k}^{i}) = k  z_{k}^{i}.
\end{gather*}
\item \label{prop_pL:item_pl_tau} The induced morphism $\pLinr(\tau^k_{k-r})\colon \pLinr(F_k)\rightarrow \pLinr(F_{k-r})$ coincides with $\pLinr(F_k) \rightarrow \pLinr(F_k)[\Delta^1_{\pLinr(F_{k})} \leq k-r-1]$, where $\tau^k_{k-r} = p^{F_k}_{[\Delta_{F_k}\leq k-r]}\colon  F_k\to F_{k-r}$ is the canonical projection on a lower graded bundle.
\item \label{prop_pL:item_injective} Given an injective graded bundle morphism $\phi\colon F_k\rightarrow F_k'$, the associated morphism $\pLinr(\phi)$ is injective, as well.
\end{enumerate}
\end{Proposition}

\begin{proof} We shall only give a geometrical explanation of~\eqref{prop_pL:itemTkM}. Another proof can be found in the arXiv preprint version of~\cite{Bruce:2014b}. The statement~\eqref{prop_pL:item_pullback} follows directly from the def\/inition of the linearisation functor as both vector bundles in interest are of the same rank. The statements~\eqref{prop_pL:item_pl_tau} and~\eqref{prop_pL:item_injective} are very easy to proof directly, while~\eqref{prop_pL:item embedding} was carefully proved in~\cite{Bruce:2014}.

Recall~\cite{Leon:1992, Eliopoulos:1966}, a $k$-th tangent bundle $\sT^k M$ is equipped with an endomorphism~$J$ of its tangent bundle, called the canonical \emph{almost tangent structure of order~$k$}, which has the form
\begin{gather*}
J \pa_{x^{A, (\alpha)}} = \pa_{x^{A, (\alpha+1)}} \quad (\text{for}\ 0\leq \alpha\leq k-1), \qquad J \pa_{x^{A, (k)}} = 0,
\end{gather*}
so the image of $J^i=\underbrace{J\circ \cdots \circ J}_{i\text{-times}}$ coincides with the kernel of $J^{k+1-i}$ for $i=1, \ldots, k$. Thus we have an exact sequence
\begin{gather*}
0\rightarrow \operatorname{im} J^k \rightarrow \sT \sT^{k} M \xrightarrow{J} \sV \sT^k M.
\end{gather*}
The image of $J^k$ is given in $\sT \sT^{k}M$ by vanishing all dif\/ferentials $d x^{A, (\alpha)}$ with $0\leq \alpha <k$, hence $\sT \tau^k_{k-1} \colon \sT \sT^k M \twoheadrightarrow \sT\sT^{k-1} M$ vanishes on the kernel of $J$, and so it factors to a vector bundle morphism $\sV \sT^k M \twoheadrightarrow \sT \sT^{k-1} M$. It is immediate to check that the obtained projection coincides with $p^{\sV F_k}_{\pLinr F_k}\colon \sV F_k \rightarrow \pLinr(F_k)$, where $F_k=\sT^k M$.
\end{proof}

The linearisation functor associates in a canonical functorial way a $\GrL$-bundle with any graded bundle. It is clear that one can now iterate the application of the linearisation functor and from a graded bundle of degree $k$ produce a $k$-fold vector bundle. We describe this construction carefully in the next section.

\section[Graded bundles vs $k$-fold vector bundles]{Graded bundles vs $\boldsymbol{k}$-fold vector bundles}\label{sec:GradedBundlesandHigherVectorBundles}
\subsection{The full linearisation functor}

Recall, that the linearisation functor $\pLinr$ takes a graded bundle $F_k$ of degree $k$ to a double graded bundle of bi-degree $(k-1,1)$. We may apply the functor $\pLinr$ again to the graded bundle structure of degree $k-1$ on $\GrL$-bundle $\pLinr(F_k)$ to get a triple graded bundle of multi-degree $(k-2,1,1)$. After $k-1$ iterations we arrive at a $k$-fold vector bundle denoted by $\pLinr^{(k-1)}(F_k)$ which we take as a~def\/inition of the \emph{full linearisation functor}.
\begin{Definition} \label{def:full_linearisation_functor} The full linearisation functor
\begin{gather*}
\Linr\colon \ \GrB[k] \rightarrow \VB[k]
\end{gather*}
is def\/ined as the composition of the linearisation functor $(k-1)$ times, i.e., $\Linr(F_k) = \pLinr^{(k-1)} (F_k)$.
\end{Definition}

\begin{Example}
$\Linr(\sT^k M) = \sT^{(k)} M$ by Proposition~\ref{prop:pLin}.
\end{Example}

Similarly, we can consider iterations of the vertical functor. Recall, the vertical functor $\sV$ takes a graded bundle of degree $k$ to a double graded bundle of degree $(k, 1)$. Applying again the functor $\sV$ to $\big(\sV(F_k), \Delta^1_{\sV F_k}\big)$ we obtain $\sV \sV F_k$ which is a $3$-fold graded bundle of degree $(k, 1, 1)$. Continuing this way we f\/ind that for any $r\geq 1$, $\sV^{(r)} F_k =
\underbrace{\sV \sV\cdots \sV}_{r\text{-times}} F_k$ is a $(r+1)$-fold graded bundle of degree $(k, 1, \ldots, 1)$.

\begin{Proposition}\label{prop:l_r_epi_and_inclusion} There is a canonical epimorphism $(0\leq r\leq k-1)$
\begin{gather}\label{prop:epimorphism_r}
p^{\sV^{(r)}F_k}_{\pLinr^{(r)}({F}_k)}\colon \ \sV^{(r)} F_k \rightarrow \pLinr^{(r)} (F_k) = \sV^{(r)} F_k \big[\Delta^1_{\sV^{(r)} F_k} \leq k-r\big]
\end{gather}
and an embedding
\begin{gather}\label{prop:embedding_r}
\pLinr^{(r)} (F_k) \hookrightarrow \sT^{(r)}\sT^{k-r} F_k.
\end{gather}
Both mappings respect $(r+1)$-fold graded bundle structures.
\end{Proposition}
\begin{proof} We shall prove \eqref{prop:epimorphism_r} by induction. The case $r=0$ is trivial and the case $r=1$ follows directly from the def\/inition of the functor $\pLinr$. Denote $\tilde{F}_k := \sV^{(r)} F_k$ which we consider now as a~graded bundle of degree $k$. By the inductive hypothesis, the base of the canonical projection $\tilde{\tau}^k_{k-r}\colon \tilde{F}_k\rightarrow \tilde{F}_{k}\big[\Delta^1_{\tilde{F}_k \leq k-r}\big]$, which is def\/ined in~\eqref{eqn:fibrations}, coincides with $\pLinr^{(r)}(F_k)$. To get $\pLinr^{(r+1)} F_k$, apply the functor $\pLinr$ to the projec\-tion~$\tilde{\tau}^k_{k-r}$, then compose it with the canonical projection~$p^{\sV \tilde{F}_k}_{\pLinr(\tilde{F}_k)}$ on the linearisation of $\tilde{F}_k$. Using Proposition~\ref{prop:pLin}(\ref{prop_pL:item_pl_tau}) we arrive at the sequence
\begin{gather*}
\sV^{(r+1)} F_k = \sV \tilde{F}_k \xrightarrow{p^{\sV \tilde{F}_k}_{\pLinr(\tilde{F}_k)}} (\sV \tilde{F}_k) \big[\Delta^1_{\sV^{(r+1)}F_k}\leq k-1\big] = \pLinr(\tilde{F}_k) \xrightarrow{\pLinr(\tilde{\tau}^k_{k-r})} \pLinr^{(r+1)} (F_k)  \\
\hphantom{\sV^{(r+1)} F_k = \sV \tilde{F}_k \xrightarrow{p^{\sV \tilde{F}_k}_{\pLinr(\tilde{F}_k)}}}{}
\overset{(*)}{=} \pLinr (\tilde{F}_k) \big[\Delta^1_{\pLinr (\tilde{F}_k)}\leq k-r-1\big] = \sV^{(r+1)} F_k \big[\Delta^1_{\sV^{(r+1)} F_k}\leq k-r-1\big]
\end{gather*}
that completes our inductive reasoning.

To prove~\eqref{prop:embedding_r} consider the graded bundle $\tau\colon F_k\rightarrow M$ as a canonical substructure of $\tau^k_N\colon \sT^k N\rightarrow N$ with $N=F_k$ (for this fundamental observation see~\cite{JG_MR_higher_vb})  and apply the functor $\pLinr$ to both structures. By Proposition~\ref{prop:pLin}\eqref{prop_pL:itemTkM} and~\eqref{prop_pL:item_injective} we get a $\GrL$-bundle embedding $\pLinr(\iota)\colon \pLinr(F_k) \rightarrow \pLinr(\sT^k N) = \sT \sT^{k-1} N$, where $\iota\colon F_k\rightarrow \sT^k N$ denotes the canonical embedding. We can iterate the application of the linearisation functor and f\/inally get~\eqref{prop:embedding_r}.
\end{proof}

\begin{Example}
Consider a graded bundle $F_{2}$ equipped with homogeneous local coordinates $(x^{A}, ~ y^{a}, ~ z^{i})$ of weight $0,1$ and $2$, respectively. Then, $\sV F_2\subset \sT F_2$, $\sV^{(2)}F_{2} \subset \sT^{(2)} F_2$, and $\sV^{(3)} F_2\subset \sT^{(3)} F_2$ come equipped with homogeneous coordinates
\begin{gather*}
\big(x^{A,(0)}_0, y^{a,(0))}_1, z^{i,(0)}_2; y^{a,(1)}_0, z^{i,(1)}_1\big), \\
\big(x^{A, (0,0)}_{0}, y^{a, (0,0)}_{1}, z^{i, (0,0)}_{2}, y^{a, (0,1)}_{0}, z^{i, (0,1)}_{1}; y^{a,(1, 0)}_{0}, z^{i, (1,0)}_{1}, z^{i, (1,1)}_{0}\big), \\
\big(x^{A, (000)}_{0}, y^{a, (000)}_1, z^{i, (000)}_2, y^{a, (001)}_0, z^{i, (001)}_1, y^{a,(010)}_0, z^{i, (010)}_1, z^{i, (011)}_0;  \\
\qquad {} y^{a,(100)}_0, z^{i,(100)}_1, z^{i,(101)}_0, z^{i,(110)}_0\big),
\end{gather*}
respectively. The f\/irst component written in the bottom of the multi-weight is that naturally inherited from $F_{2}$, while the other components come from the vector bundle structure of the tangent bundles. In full generality, consider a graded bundle $F_{k}$ equipped with homogeneous local coordinates $\big(x^{A}, y^i_w\big)$, then $\sV^{(r)}F_{k}$ can be equipped with inherited from~$\sT^{(r)} F_k$ homogeneous local coordinates
\begin{gather*}
\big(x^A, y^{i, (\ep)}_w\big),
\end{gather*}
where $\veps\in\{0,1\}^{r}$ and $|\ep|\leq w$. The weight of $y^{i, (\ep)}_w$ with respect to $\Delta^1_{\sV^{(r)}} F_k$ is $w - |\ep|$ (not explicitly written now). Consider $(\sT^{(r)} F_k; \Delta^0, \Delta^1, \ldots, \Delta^r)$ as a $(r+1)$-fold graded bundle with Euler vector f\/ields $\Delta^1, \ldots, \Delta^r$ of $r$-th iterated tangent bundle of $F_k$ (namely, $\Delta^i$ is the Euler vector f\/ield of the vector bundle projection $\sT^{(i-1)} \tau_{T^{(r-i)} F_k}\colon \sT^{(r)} F_k \rightarrow \sT^{(r-1)} F_k$, $1\leq i\leq r$) and the weight vector f\/ield $\Delta^0:= \dd_{\sT}^{(r)}\Delta_{F_k}$ being the (iterated) tangent lift of the weight vector f\/ield~$\Delta_{F_k}$ of~$F_k$. Def\/ine~$X_r \in \mathfrak{X}(\sT^{(r)} F_k)$ as $X_r = \Delta^0 - \sum_1^k \Delta^k$, so the weight of $y^{a, (\ep_1, \ldots, \ep_r)}_w$ with respect to $X_r$ is $w-(\ep_1+\cdots+\ep_r)$, hence
 \begin{gather}\label{eqn:V_r}
 \sV^{(r)} F_k = \sT^{(r)}F_k[X_r\geq 0].
 \end{gather}
Note that $\Delta^{1}_{\sV^{(r)}F_{k}} = X_{r}|_{\sV^{(r)}F_{k}}$.
\end{Example}

\begin{Example}\label{exmpl:F3}
Let us consider in some detail the generic case when we are dealing with a graded bundle of order three. This example should be enough to highlight the main constructions explicitly without too much notational clutter. Let us employ homogeneous local coordinates $\big(x^{A}, y^{a}, z^{i},w^{\kappa}\big)$ on $F_{3}$ where the weight is assigned as $0$, $1$, $2$ and $3$, respectively. The admissible changes of local coordinates are of the form
\begin{gather*}
 x^{A'} = x^{A'}(x), \qquad  y^{a'} = y^{b}T_{b}^{\:\:\: a'}(x), \qquad  z^{i'} = z^{j}T_{j}^{\:\:\: i'}(x) + \frac{1}{2}y^{a}y^{b}T_{ba}^{\:\:\:\: i'}(x),\\
w^{\kappa'} = w^{\mu}T_{\mu}^{\:\: \kappa'}(x) + z^{i}y^{a}T_{ai}^{\:\:\: \kappa'}(x) + \frac{1}{3!}y^{a}y^{b}y^{c}T_{cba}^{\:\:\:\:\:\: \kappa'}(x). \end{gather*}

Then on $\pLinr(F^3)$ we can employ coordinates $(x^A, \dot{y}^a, y^a, \dot{z}^i,z^i, \dot{w}^\kappa)$ of weights $(0,0)$, $(0,1)$, $(1, 0)$, $(1, 1)$, $(2,0)$ and $(2, 1)$, respectively. The additional transformation rules are
\begin{gather*}
\dot{y}^{a'} = \dot{y}^{b} T^{a'}_{b}, \qquad \dot{z}^{i'}=\dot{z}^j T^{i'}_j + \dot{y}^a y^b T^{i'}_{ba}, \\
\dot{w}^{\kappa'} = \dot{w}^{\mu}T_{\mu}^{\:\: \kappa'} + (\dot{z}^{i}y^{a} + z^{i}\dot{y}^{a}) T_{ai}^{\:\:\: \kappa'} + \frac{1}{2!} \dot{y}^{a}y^{b}y^{c}T_{cba}^{\:\:\:\:\:\: \kappa'}.
\end{gather*}
 By iterating this construction and dif\/ferentiation above formulas, we get $\pLinr^{(2)} F_3$ with coordinates
 $(x^A$, $\dot{y}^a=y^{a,(10)}$, $y^a=y^{a, (00)}$, $\dot{z}^i=z^{i,(10)}$, $dy^a=y^{a,(01)}$, $d\dot{z}^i=z^{i,(11)}$, $dz^i=z^{i,(01)}$, $d\dot{w}^\kappa=w^{\kappa,(11)})$ of degrees $(000)$, $(010)$, $(100)$, $(110)$, $(001)$, $(011)$, $(101)$, $(111)$, respectively, and transformation laws for extra coordinates
\begin{gather*}\allowdisplaybreaks
dy^{a'} = dy^b T^{a'}_b(x), \\
d\dot{z}^{i'} = d\dot{z}^j  T^{i'}_j(x) + \dot{y}^a dy^b T^{i'}_{ba}(x),\\
d z^{i'} = d z^j T^{i'}_j(x) + dy^a y^b T^{i'}_{ba}(x),\\
d \dot{w}^{\kappa'} = d \dot{w}^\mu\,T^{\kappa'}_\mu(x) + (\dot{z}^i dy^a + d\dot{z}^i y^a + dz^i \dot{y}^a)T^{\kappa'}_{ai}(x) +
\dot{y}^a\,dy^b y^c T^{\kappa'}_{cba}(x).
\end{gather*}
We see that transformation laws for coordinates of the same total degree are essentially the same. We conclude that $\pLinr^{(2)}(F_3)$ is a triple vector bundle with isomorphic feet $\pLinr(F_2)$
\begin{gather*}
\xymatrix{
& \pLinr F_2 \ar[rr]\ar[dd] && F_1\ar[dd] \\
\Linr F_3 = \pLinr^{(2)} F_3 \ar[dd] \ar[ru]\ar[rr] && \pLinr F_2 \ar[ru]\ar[dd] & \\
& F_1 \ar[rr]&& M \\
\pLinr F_2 \ar[ru] \ar[rr]&& F_1 \ar[ru] &}
\end{gather*}
\end{Example}

In full generality, it is natural to use the coordinate system $\big(x^A, y^{a, (\ep)}_w\big)$, $\ep\in\{0,1\}^r$, for $\pLinr^{(r)} F_k$ as well (now we take only such $y^{a, (\ep)}_w$ that $0\leq w-|\ep| \leq k-r$), as these coordinates are constant on f\/ibers of the f\/ibration $p^{\sV^{(r)}(F_k)}_{\pLinr^{(r)}(F_k)}$.

On the other hand, consider $F_3$ as a substructure of $\sT^3 F_3$ given by equations $x^{A, (\beta)} = 0$, $y^{a, (\alpha)}_w = 0$, for $1\leq \alpha, \beta\leq 3$ such that $\alpha \neq w$. Then we recognise $\pLinr(F_3)$ as a subset of $\sT \sT^2 F_3$ and $\Linr(F_3)=\pLinr^{(2)}(F_3)$ as a subset of $\sT^{(3)} F_3$, the latter is def\/ined by equations $x^{A, (\ep)} =0$, for $\ep\neq (0,0,0)$ and $y_w^{a, (\ep)} = 0$ if $\ep_1+\ep_2+\ep_3\neq w$ in local coordinate system $\big(x^A, y_w^{a, (\ep)}\big)$, $\ep = (\ep_1, \ep_2, \ep_3)\in\{0, 1\}^3$, on $\sT^{(3)} F_3$. Transformations for the remaining coordinates are the same as given above for $\pLinr^{(2)} F_3$ with the identif\/ication $\big(x^A, y_{|\ep|}^{(\ep_1, \ep_2, \ep_3)}\big) \mapsto \big(x^A, y_{|\ep|}^{(\ep_2, \ep_3)}\big)$, $|\ep| = \ep_1+\ep_2+\ep_3$. The weight of $y_w^{a, (\ep)}$ with respect to the vector f\/ield $X_3\in \mathfrak{X}(F_3)$, def\/ined as the dif\/ference between $\dd_\sT^{(3)} \Delta_{F_3}$ and the sum of three Euler vector f\/ields def\/ining triple vector bundle structure on $\sT^{(3)} F_3$, is equal $w-(\ep_1+\ep_2+\ep_3)$. Thus $\Linr(F_3)$ as a subset of $\sT^{(3)} F_3$ is the set of singular points of the vector f\/ield~$X_3$.

In general, the embedding $\Linr (F_k) \subset \sT^{(k)} F_k$ def\/ined in~\eqref{prop:embedding_r} can be given an explicit form, namely
\begin{gather}\label{eqn:full_linearisation_def}
\Linr(F_k) \simeq \sT^{(k)}F_k[X_k=0],
\end{gather}
where $X_k = \Delta^0 - \sum_1^k \Delta^k$ is as above. (It is not worth to consider $\Linr(F_k)$ as a $(k+1)$-fold graded bundle since $\Delta^0$ restricted to $\Linr(F_k)$ coincides with the sum of Euler vector f\/ields $\Delta^i$, so $(\sT^{(k)} F_k [X = 0]; \Delta^{0})$ is just a graded bundle of degree $k$ def\/ined by taking the total degree in $k$-fold vector bundle~$\Linr(F_k)$.)

Similarly, the epimorphism $\sV^{(k-1)} F_k \rightarrow \Linr(F_k)$ has the explicit form
\begin{gather}\label{eqn:LF_k_explicit}
\sV^{(k-1)} F_k \rightarrow \big(\sV^{(k-1)}F_{k} \big)\big[\Delta^{1}_{\sV^{(k-1)}F_{k}} \leq 1\big]
\end{gather}
and $\Delta^{1}_{\sV^{(k-1)}F_{k}} = X_{k-1}|_{\sV^{(k-1)}F_{k}}$. The isomorphism $\big(V^{(k-1)}F_{k} \big)[\Delta^{1}_{V^{(k-1)}F_{k}} \leq 1] \simeq \sT^{(k)}F_k[X=0]$ has the form
\begin{gather}\label{eqn:local_iso_LFk_explicit}
\big(x^A, y^{i, (\ep)}_w\big) \mapsto \big(x^A, y^{i, (w-|\ep|, \ep_2, \ldots, \ep_k)}_w\big),
\end{gather}
where $\ep = (\ep_2, \ldots, \ep_k)\in \{0,1\}^{k-1}$, so $|\ep| = \ep_2+ \cdots+ \ep_k$, and $0\leq w-|\ep|\leq 1$.

\begin{Corollary}\label{cor:L_emb_and_epim} The full linearisation of a graded bundle $F_k$ can be seen as a quotient structure of the iterated vertical bundle $\sV^{(k-1)} F_k$ and as a substructure of the iterated tangent bundle~$\sT^{(k)} F_k$, as well. The local coordinate formulas are obtained from~\eqref{eqn:V_r}, \eqref{eqn:full_linearisation_def}, \eqref{eqn:LF_k_explicit} and~\eqref{eqn:local_iso_LFk_explicit}.
\end{Corollary}

In view of Corollary~\ref{cor:L_emb_and_epim}, it does not fundamentally matter if we chose to interpret the induced homogeneous coordinates on $\Linr(F_{k})$ as coming from $\sV^{(k-1)}F_{k}$ or $\sT^{(k)}F_{k}$, though the latter maybe easier to comprehend. However, when examining the structure of the full linearisation, understanding~$\Linr(F_{k})$ as a substructure of~$\sT^{(k)}F_{k}$ will prove useful.

\subsection{Characterisation of the full linearisation for the degree two case}\label{ssec:Characterisation_deg_two}
The full linearisation is a functor from the category of graded bundles of degree $k$ to the category of $k$-fold vector bundles, however it is clear that not all $k$-fold vector bundles arise as the full linearisation of a graded bundle. All $k$-fold vector bundles are isomorphic, but usually non-canonically, to a \emph{decomposed} or \emph{split} $k$-fold vector bundle which is build out of $2^k-1$ vector bundles $E_\ep \rightarrow M$, where the weight $\ep\in \{0,1\}^k$ is a nonzero sequence, known as \emph{building $($vector$)$ bundles.} Essentially, every $k$-fold vector bundle is built from its building bundles, just not in a canonical fashion. For example, for any DVB, say $\sD = (D; A, B; M)$ with the core bundle $C\rightarrow M$, there exist a double vector bundle isomorphism $\sD \rightarrow (A\times_{M} B\times_{M} C; A, B; M)$ inducing the identity on all the building vector bundles $A=\sD_{(0,1)}$, $B=\sD_{(1,0)}$ and $C=\sD_{(1,1)}$.

If the $k$-fold vector bundle in question is balanced, then every vector bundle $E_\ep$ of a given total weight are identif\/ied and thus \emph{any} balanced $k$-fold vector bundle is the full linearisation of some graded bundle of degree~$k$. However, this is not suf\/f\/icient to completely characterise the full linearisation. There are morphisms between the full linearisations that are permissible and not permissible; the latter means that they cannot be traced back to morphisms between the original graded bundles. For concreteness let us f\/irst consider the degree two case carefully.

The full linearisation $\sD := \Linr(F_{2})$ of a graded bundle $F_{2}$ comes with an additional canonical structure: a double vector bundle morphism $\sigma\colon \sD \rightarrow \sD^{\f}$, where $\sD^{\f}=(D,\zD^2,\zD^1)$ is the f\/lip of $\sD=(D,\zD^1,\zD^2)$. The morphism $\sigma$ simply exchanges the legs (so the Euler vector f\/ields), and can be def\/ined as the restriction of the canonical involution $\kappa\colon \sT\sT F_{2} \rightarrow \sT\sT F_{2} $ to the double vector subbundle $D \subset \sT\sT F_2=\sT^{(2)}F_{2}$. In local coordinates as introduced earlier, the canonical DVB morphism is given by
\begin{gather} \label{eqn:sigma_ab_i}
 \sigma^{*}y^{a,(10)} = y^{a,(01)}, \qquad  \sigma^{*}y^{b,(01)} =y^{b,(10)},\qquad \sigma^{*}z^{i,(11)} = z^{i,(11)}.
\end{gather}

\begin{Remark} One must note that a general balanced DVB does \emph{not} admit a canonical $\sigma\colon$ $\sD \rightarrow \sD^{\f}$, the local description above is given in specif\/ic (adapted) coordinates on $\sT^{(2)} F_2$.
\end{Remark}

Now we can consider a category whose objects are double vector bundles $\sD=(D; A, B; M)$ equipped with a DVB morphism $\sigma\colon \sD\rightarrow \sD^\f$, such that
\begin{gather}\label{eqn:condition_on_sigma1}
\sigma^\f\circ \sigma = \id_D
\end{gather}
and
\begin{gather}\label{eqn:condition_on_sigma2}
\sigma|_C = \id_C,
\end{gather}
where $C\rightarrow M$ denotes the core bundle of $\sD$. It follows from \eqref{eqn:condition_on_sigma1} that $\sD$ is balanced, as $\sigma|_A\colon A\rightarrow B$ and $\sigma|_B\colon B\rightarrow A$ are vector bundle morphisms satisfying $\sigma|_B\circ \sigma|_A = \id_A$ and $\sigma|_A\circ \sigma|_B = \id_B$. Under the identif\/ication $A = B$ given by~$\sigma|_A$, it follows that the following sets coincide
\begin{gather*}
\{q\in D\colon \sigma(q)=q\} = \big\{q\in D\colon \pi^D_A(q) = \pi^D_B(q)\big\}.
\end{gather*}

\begin{Definition} Such pairs $(\sD, \sigma)$ we will refer to as \emph{symmetric double vector bundles} and the above set as the \emph{diagonal} or the set of \emph{holonomic vectors} of $(\sD,\zs)$. A morphism of symmetric double vector bundles $(\sD, \sigma) \rightarrow (\sD', \sigma')$ is a DVB morphism $\phi\colon \sD\rightarrow \sD'$ commuting with the additional structure:
\begin{gather}\label{eqn:morphism}
\sigma' \circ \phi = \phi^\f \circ \sigma.
\end{gather}
As the f\/lip operation is a functor, $\phi^{\f} \colon \sD^{\f} \rightarrow \sD'^{\f}$ coincides with $\phi$ as a map between sets; a DVB and its f\/lip have the same underlying manifold structure. Composition of morphisms in this category is the standard composition of smooth maps between manifolds. We will denote this category as~$\SymmVB[2]$.
\end{Definition}

Let us pick local coordinates $\big(x^{A}, y_{(1,0)}^{a}, y_{(0,1)}^{b}, z_{(1,1)}^{i}\big)$ on an arbitrary symmetric double vector bundle $(\sD, \sigma)$ in such a way that $\sigma^{*} y^{a}_{(0,1)} = y^{a}_{(1,0)}$, so $\sigma^{*}y^{a}_{(1,0)} = y^{a}_{(0,1)}$ by~\eqref{eqn:condition_on_sigma1}. Then, by \eqref{eqn:condition_on_sigma2}, the isomorphism $\sigma\colon \sD \rightarrow \sD^{\f}$ is of the form
\begin{gather*}
 \sigma^{*} z^{i}_{(1,1)} = z^{i}_{(1,1)} + y^{a}_{(1,0)}y^{b}_{(0,1)}\sigma_{ba}^{i}(x).
\end{gather*}
From \eqref{eqn:condition_on_sigma1}, `applying sigma twice' to the $z$ coordinate results in
\begin{gather*}(\id_{D})^{*}z^{i}_{(1,1)} = z^{i}_{(1,1)} + y^{a}_{(1,0)}y^{b}_{(0,1)}\sigma_{ba}^{i}(x) + y^{a}_{(0,1)}y^{b}_{(1,0)}\sigma_{ba}^{i}(x), \end{gather*}
and thus $\sigma_{ba}^{i} = {-} \sigma_{ab}^{i}$.

A morphism of double vector bundles $\phi\colon \sD \rightarrow \sD'$ is locally of the form (we do not explicitly write the morphism $M \rightarrow M'$ as its exact form is not important for our discussion)
\begin{gather*}
\phi^{*} y^{a'}_{(1,0)} = y^{a}_{(1,0)} Q_{a}^{\:\:\: a'}(x), \qquad \phi^{*} y^{b'}_{(0,1)} = y^{b}_{(0,1)} R_{b}^{\:\:\: b'}(x),\\
 \phi^{*} z^{i'}_{(1,1)} = z^{i}_{(1,1)}Q_{i}^{\:\:\: i'}(x) + y^{a}_{(1,0)} y^{b}_{(0,1)}Q_{ba}^{\:\:\:\: i'}(x).
\end{gather*}
The reader can quickly convince themselves that if we further insist on the condition (\ref{eqn:morphism}), then $Q_{a}^{\:\: \: a'} = R_{a}^{\:\: \: a'}$ and $Q_{ba}^{\:\:\:\: i'} = Q_{ab}^{\:\:\:\: i'}$. We now explicitly see the role of the morphism $\sigma$ in ensuring that morphism are `symmetric'. In particular, any morphism satisfying (\ref{eqn:morphism}) comes from a~morphism of graded bundles $F_{2} \rightarrow F_{2}'$, where $F_{2} := \{ q \in D ~ | ~ \pi^{D}_{A}(q) = \pi^{D}_{B}(q) \}$ and similar for~$F_{2}'$.

As a matter of formality, let us denote the full linearisation
\begin{gather*}
\Phi\colon \  \GrB[2] \rightarrow \SymmVB[2],
\end{gather*}
which assigns $(\Linr(F_{2}) , \kappa|_{\Linr(F_{2})})$ to $F_{2}$, and the \emph{diagonalisation} functor as
\begin{gather*}
\Psi \colon \  \SymmVB[2] \rightarrow \GrB[2],
\end{gather*}
which assigns $F_{2}:= \{ q \in D \,|\, \pi^{D}_{A}(q) = \pi^{D}_{B}(q)\}$ to $D$. It is clear that $F_{2} \simeq \Psi(\Phi(F_{2}))$ canonically.

The other way round, we claim that $\Phi \circ \Psi \simeq \textnormal{Id}_{\SymmVB[2]}$. This claim boils down to f\/inding a~canonical isomorphism
\begin{gather*}
I\colon \  (\sD, \sigma) \rightarrow (\Linr(F_{2}), \kappa|_{\Linr(F_{2})}),
\end{gather*}
where $F_2=\Psi(\sD, \sigma)$, such that the following diagram is commutative:
\begin{gather*}
\leavevmode
\begin{xy}
(0,20)*+{\sD}="a"; (30,20)*+{\sD^{\f}}="b";%
(0,0)*+{\Linr(F_{2})}="c"; (30,0)*+{ (\Linr(F_{2}))^{\f}.}="d";%
{\ar "a";"b"}?*!/_3mm/{\sigma};
{\ar "a";"c"}?*!/^3mm/{I};
{\ar "b";"d"}?*!/_3mm/{I^{\f}};
{\ar "c";"d"}?*!/^3mm/{\kappa|_{\Linr(F_{2})}};
\end{xy}
\end{gather*}
Let us denote the naturally inherited coordinates on $F_{2} = \Psi(\sD, \sigma)$ as
 \begin{gather*}\big(x^{A},\, y^{a} = y^{a}_{(1,0)}|_{F_{2}} = y^{a}_{(0,1)}|_{F_{2}},\, z^{i} = z^{i}_{(1,1)}|_{F_{2}}\big)\end{gather*}
 and on the full linearisation $\Linr(F_{2}) \subset \sT^{(2)}F_{2}$ we have the induced coordinates
 \begin{gather*}\big(x^{A},\, y^{a,(1,0)}, \, y^{b,(0,1)}, \, z^{i,(1,1)}\big).\end{gather*}
It is then a matter of calculation (we omit details) to show that
\begin{gather}\label{eqn:def_iso_I}
 I^{*} y^{a,(1,0)} = y^{a}_{(1,0)}, \qquad  I^{*} y^{b,(0,1)} = y^{b}_{(0,1)}, \qquad I^{*} z^{i,(1,1)} = z^{i}_{(1,1)} + \frac{1}{2} y^{a}_{(1,0)} y^{b}_{(0,1)}\sigma_{ba}^{i}(x)
\end{gather}
is the natural isomorphism we are looking for. The reader should note that $I$ does not actually depend on the coordinates employed on $D$. To sum up, we have thus established the following:
\begin{Proposition}\sloppy
There is an equivalence of categories between $\GrB[2]$ and the category $\SymmVB[2]$.
\end{Proposition}

\begin{Example}
Consider $D = \sT^{(2)}M$ and $ \sigma = \kappa_M\colon \sT^{(2)}M \rightarrow \sT^{(2)}M$ which is not an automorphism but an isomorphism
 from $\sT^{(2)}M$ to the f\/lip $(\sT^{(2)}M)^{\f}$. Thus $(\sT^{(2)}M, \kappa_M)$ is a symmetric DVB and the diagonalisation is clearly $\sT^{2}M$.
\end{Example}

\subsection{Relation with symplectical double vector bundles}
Consider a symmetric double vector bundle $(\sD, \sigma)$, where $\sD = (D; A, B; M)$, so $\sigma\colon (D; A, B; M)$ $\rightarrow (D; B, A; M)$ is a DVB morphism satisfying~\eqref{eqn:condition_on_sigma1} and~\eqref{eqn:condition_on_sigma2}. Taking the dual of $\sigma$ with respect to vector bundle structures def\/ined by vertical arrows produces a~\emph{map} $D_{B}^{*} \rightarrow D^{*}_{A}$, and not just a~relation, as $\sigma$ is an isomorphism. We shall see soon that it is a DVB morphism depicted in the diagram~\eqref{diag:sigma_V} which we shall denote by~$\sigma^\sv\colon \sD^{\f\sv} \rightarrow \sD^\sv$. Moreover, $\sigma^{\sv}|_{C^*} = \id_{C^*}$ because of~\eqref{eqn:condition_on_sigma2},
\begin{gather}\label{diag:sigma_V}
\leavevmode
\begin{tabular}{ >{\centering\arraybackslash}m{0.5in} >{\centering\arraybackslash}m{0.5in} >{\centering\arraybackslash}m{0.5in} >{\centering\arraybackslash}m{0.5in} >{\centering\arraybackslash}m{0.5in}}
\begin{xy}
(0,20)*+{(D^{*}_{B})^*_{C^*}}="a"; (20,20)*+{C^{*}}="b";
(10, 10)*+{\sE^\sh}="e";
(0,0)*+{A}="c"; (20,0)*+{ M}="d";%
{\ar "a";"b"};
{\ar "a";"c"};
{\ar "b";"d"};
{\ar "c";"d"};
\end{xy}
&
$\stackrel{I}{ \xleftarrow{\hspace*{1cm}} }$
&
\begin{xy}
(0,20)*+{D^{*}_{A}}="a"; (20,20)*+{C^{*}}="b";
(10, 10)*+{\sD^\sv}="e";
(0,0)*+{A}="c"; (20,0)*+{ M}="d";%
{\ar "a";"b"};
{\ar "a";"c"};
{\ar "b";"d"};
{\ar "c";"d"};
\end{xy}
&
$\stackrel{\sigma^{\sv}}{ \xleftarrow{\hspace*{1cm}} }$
&
\begin{xy}
(0,20)*+{D^{*}_{B}}="a"; (20,20)*+{C^{*}}="b";
(10, 10)*+{\sE=\sD^{\f\sv}}="e";
(0,0)*+{B}="c"; (20,0)*+{ M.}="d";%
{\ar "a";"b"};
{\ar "a";"c"};
{\ar "b";"d"};
{\ar "c";"d"};
\end{xy}
\end{tabular}
\end{gather}

The vector bundles $D^{*}_{A}$ and $D^{*}_{B}$ are in duality~\cite{Konieczna:1999} (over base $C^{*}$, where $C$ is the core of $\sD$), which results in a DVB isomorphism $\sD^\sv = (D^*_A; A, C^*; M)\rightarrow \sD^{\f\sv\sh} = ((D^*_B)^*_{C^*}; A, C^*; M)$. It is worth to remember that with the choice of the pairing
\begin{gather}\label{eqn:DVB_pairing}
 \dblInnerBracket{\cdot, \cdot}\colon \  D^*_A\times_{C^*} D^*_B \rightarrow \R, \qquad \dblInnerBracket{\Phi, \Psi} = \Phi(d)-\Psi(d), \qquad
 d\in D,
\end{gather}
as given in~\cite{Konieczna:1999}, where $\Phi\in (D_a)^*$, $\Psi\in (D_b)^*$ and $d$ maps to $(a, b)$ under the canonical projection $D\rightarrow A\times_{M} B$, the DVB isomorphism $I\colon \sD^\sv \simeq \sD^{\f\sv\sh}$ is the identity on the bundle $C^*\rightarrow M$, and the core bundle $B^*\rightarrow M$, but it is minus the identity on the bundle $A\rightarrow M$. Moreover, there is no canonical DVB isomorphism which is the identity on all the building bundles. Nevertheless, $\sigma^{\sv}$ gives a vector bundle isomorphism from~$D^{*}_{B}$ to its dual over $C^*$ covering $\id_{C^*}$, and so the additional structure of $\sigma$ gives rise to a non-degenerate bilinear form on $E\rightarrow C^*$, where \mbox{$E=D^*_B$}. We shall shortly show that the bilinear form is skew-symmetric, and it induces a DVB isomorphism from $\sE = \sD^{\f\sv} = (E; B, C^*; M)$ to $E^\sh$. These observation lead us to the following def\/inition.

\begin{Definition}
A \emph{symplectical double vector bundle} is a double vector bundle $\sE = (E; E_{01}$, $E_{10}; M)$ equipped with a linear non-degenerate skew-symmetric form on $ E^{*}_{E_{10}} \rightarrow C^{*}$ that induces a double vector bundle isomorphism $\beta\colon \sE \rightarrow \sE^\sh$.
\end{Definition}

A morphism between symplectical DVBs
\begin{gather*}
\Omega\colon \  \sE = (E; E_{01}, E_{10}; M) \dashrightarrow \sE' = (E'; E'_{01}, E'_{10}; M')
\end{gather*}
 is a relation of the form $\Omega = \gamma^\sv$ for some DVB morphism
 \begin{gather*}
 \gamma\colon \ \sE'^\sv = ({E'}^*_{E'_{01}}; E'_{01}, {E'}^*_{11}; M') \rightarrow \sE^\sv = ({E}^*_{E_{01}}; E_{01}, {E}^*_{11}; M).
\end{gather*}
It follows that $\Omega$ is a vector subbundle of $\overline{E}\times E' \rightarrow E_{11}\times E_{11}'$. It is required that it is an isotropic subbundle with respect to the dif\/ference of the skew-symmetric forms on $E'$ and $E$ (which we indicated by the bar over $E$).

Evidently, symplectical double vector bundles form a category with we denote as $\SpVB[2]$.

\begin{Remark}
We will use the nomenclature \emph{symplectical double vector bundle} as opposed to \emph{symplectic double vector} as to avoid possible confusion with double vector bundles equipped with symplectic forms (of whatever weight) on their total spaces. This choice also avoids confusion with $n$-graded symplectic manifolds as def\/ined and studied in~\cite{Grabowski:2006}.
\end{Remark}

To highlight the main constructions of this subsection, given any symmetric double vector bundle $(\sD, \sigma)$ one associates with it a sympletical double vector bundle $(D_{B}^{*}, \langle , \rangle, \sigma^{\sv})$, where the pairing is given by
\begin{gather*}\langle \Psi_{1}, \Psi_{2} \rangle:= \Psi_{1}(\sigma(d)) - \Psi_{2}(d),\end{gather*}
where $\Psi_{1}, \Psi_{2} \in D_{B}^{*}$ project to the same point in $C^{*}$. The point $d\in D$ is any point that projects to $(a_{1},b_{2}) \in A\times_{M}B$, we def\/ine $b_{2} := \pi_{B}^{D_{B}^{*}}(\Psi_{2})\in B$ and $a_{1} := \sigma(\pi_{B}^{D_{B}^{*}}(\Psi_{1})) \in A$. This association is functorial, moreover we have the following theorem.

\begin{Theorem}\label{thm:equivalence_symmetric-symplectical}
The category of symplectical double vector bundles is equivalent to the category of symmetric double vector bundles, and hence is also equivalent to the category of graded bundles of degree~$2$.
\end{Theorem}

We precede the proof with a technical lemma.
\begin{Lemma} Assume $R\colon \sD'=(D'; A', B'; M') \dashrightarrow \sD=(D; A, B; M)$ is a DVB relation, i.e., $R\subset D'\times D$ is a double vector subbundle. Then $R^\sv\colon \sD^\sv \dashrightarrow {\sD'}^\sv$, where
\begin{gather*}R^\sv:= \big\{(\Phi, \Phi')\in {D}^*_{A}\times {D'}^*_{A'}: \Phi(d)=\Phi'(d')\  \text{for all}\  d, d'\in D \  \text{such that}\  (d', d)\in R\big\},
\end{gather*}
is a DVB relation, as well.
\end{Lemma}
\begin{proof} It amounts to check that $R^\sv$ is a vector subbundle of the product of vector bundles $D^*_A\rightarrow C^*$ and ${D'}^*_{A'} \rightarrow {C'}^*$. Due to \cite[Theorem 2.3]{JG_MR_higher_vb}, since $R^\sv$ is a submanifold of $D^*_A \times {D'}^*_{A'}$ it is enough to show that
\begin{gather*}
(\lambda \cdot_{C^*}\Phi, \lambda \cdot_{{C'}^*}\Phi') \in R^\sv \qquad \text{for any}\quad (\Phi, \Phi')\in R^\sv,\quad \text{and}\quad \lambda\in \R.
\end{gather*}
We have $(\lambda\cdot_{C^*} \Phi)(\lambda\cdot_{B} d) = \lambda \cdot \Phi(d) = \lambda \cdot \Phi'(d') = (\lambda\cdot_{{C'}^*}\Phi') (\lambda\cdot_{B'} d')$, hence the result follows as $(d', d)\in R$ if and only $(\lambda\cdot_{B'} d', \lambda\cdot_{B} d)\in R$.
\end{proof}
\begin{proof}[Proof of Theorem \ref{thm:equivalence_symmetric-symplectical}]
Assume $(\sD, \sigma)$ is a symmetric DVB. Then $\sigma^\sv$ is DVB isomorphism and, as we noticed earlier, it def\/ines a bilinear form on $E\rightarrow C^*$, where $E=D^*_B$. We shall show f\/irst that it is skew-symmetric. Take $\Psi_1, \Psi_2\in D_B^*$ that project to the same $\alpha\in C^*$. Denote $b_j = \pi^{D_B^*}_B(\Psi_j)\in B$, $j=1,2$, and set $a_j=\sigma(b_j)\in A$. Thus,
\begin{gather*}
\alpha(c)=\Psi_1(d_1 {+}_{B} c)-\Psi_1(d_1) = \Psi_2(d_2 {+}_{B} c)-\Psi_2(d_2)
\end{gather*}
for any $c\in C_m$, $m=\pi^B_M(b_1)=\pi^B_M(b_1)\in M$, and any $d_j\in D$ such that $\pi^D_B(d_j)=b_j$, $j=1,2$.
The bilinear form induced by $\sigma$ is given by
\begin{gather}\label{eqn:skew_form_on_E}
\langle \Psi_1, \Psi_2\rangle_E = \dblInnerBracket{\sigma^V(\Psi_1), \Psi_2} = \sigma^\sv(\Psi_1)(d)-\Psi_2(d) = \Psi_1(\sigma(d)) - \Psi_2(d),
\end{gather}
where $d$ is any element of $D$ that projects to $(a_1, b_2)$ under the projection $D\rightarrow A\times_{M} B$, and $\dblInnerBracket{\cdot, \cdot}$ is the pairing \eqref{eqn:DVB_pairing}. We have
\begin{gather*}
\langle \Psi_2, \Psi_1\rangle = \Psi_2(\sigma(d')) - \Psi_1(d'),
\end{gather*}
where now $d'\in D$ is arbitrary that projects to $(a_2, b_1)$, so we may take $d'=\sigma(d)$ due to \eqref{eqn:condition_on_sigma1} and get
\begin{gather}\label{eqn:skew-symmetry}
\langle \Psi_2, \Psi_1\rangle = \Psi_2(\sigma(\sigma(d))) - \Psi_1(\sigma(d))= - \langle \Psi_1, \Psi_2\rangle,
\end{gather}
so $\InnerBracket_E$ is skew-symmetric, as we claimed.

Assume now that $\phi\colon (\sD', \sigma') \rightarrow (\sD, \sigma)$ is a morphism in $\SymmVB[2]$. The graph of $\phi^\sv$ is not only a vector subbundle of $D_B^*\times {D'}_{B'}^*$ but is also a double vector subbundle of $\sD^\sv\times {\sD'}^\sv$. In particular, it is a vector subbundle of the product of the vector bundles $E\rightarrow C^*$ and $E'\rightarrow C^{'*}$. We need to show that it is an isotropic subbundle. Assume $(\Psi_j, \Psi_j')\in \operatorname{graph} \phi^\sv$, $j=1,2$, that project to $(b_j, b_j')\in B\times B'$ are in the same f\/iber over $C^*\times {C'}^*$. Then we can calculate
\begin{gather*}
\langle (\Psi_1, \Psi_1'), (\Psi_2, \Psi_2')\rangle_{E\times \bar{E'}} =
\langle \Psi_1, \Psi_2\rangle_E - \langle \Psi'_1, \Psi'_2\rangle_E' \\
\hphantom{\langle (\Psi_1, \Psi_1'), (\Psi_2, \Psi_2')\rangle_{E\times \bar{E'}}}{} =
\Psi_1(\sigma(d)) - \Psi_2(d) - \Psi'_1(\sigma'(d'))+ \Psi'_2(d')   \\
\hphantom{\langle (\Psi_1, \Psi_1'), (\Psi_2, \Psi_2')\rangle_{E\times \bar{E'}}}{}
=\Psi_1(\sigma(\phi(d'))) - \Psi'_1(\sigma'(d')) + \Psi'_2(d') - \Psi_2(\phi(d')) = 0,
\end{gather*}
as $\sigma\circ \phi = \phi\circ \sigma'$ and $\Psi_j' = \Psi_j\circ\phi$.

Conversely, assume that the graph of $\phi^\sv$ is an isotropic subbundle. Then, the calculation
above shows that for any $d'\in D'$ and $\Psi_1\in D^*_B$ we have
\begin{gather*}
\Psi_1(\sigma(\phi(d'))) = \Psi'_1(\sigma'(d')) = \Psi_1(\phi(\sigma'(d'))),
\end{gather*}
so $\sigma\circ \phi = \phi\circ \sigma'$. Thus we have established a functor $\SymmVB[2]\rightarrow \SpVB[2]$.

Given a symplectical DVB, say $\sE = (E; E_{01}, E_{10}; M)$, whose skew-symmetric form is def\/ined by a DVB morphism $\beta\colon \sE \rightarrow \sE^\sh$, we come back to $(\sD, \sigma)$ in an obvious way, by setting $\sigma := (I^{-1}\circ \beta)^V$ which is a DVB morphism $\sD\rightarrow \sD^\f$ satisfying \eqref{eqn:condition_on_sigma1} and \eqref{eqn:condition_on_sigma2}. Indeed, $\sigma\colon D_A\rightarrow D_B$ is characterized by the property
 \begin{gather}\label{eqn:sigma_from_beta}
 \beta(\Phi(d)) = \Phi(\sigma(d)),
 \end{gather}
 for any $\Phi\in D_B^*$ that projects to $b:=\pi^D_B(d)$. Since $\Phi$ and $\beta(\Phi)$ lies over the same functional $\alpha\in C^*$, we f\/ind from~\eqref{eqn:sigma_from_beta} that $\alpha(c) = \alpha(\sigma(c))$, hence~\eqref{eqn:condition_on_sigma2}. To see~\eqref{eqn:condition_on_sigma1}, take $\Psi_1, \Psi_2$, so that
 \begin{gather*}
 \langle \Psi_1, \Psi_2\rangle_E = \dblInnerBracket{\beta(\Psi_1), \Psi_2} = \beta(\Psi_1)(d)-\Psi_2(d),
 \end{gather*}
 where $d\in D$ is chosen so that it projects to $(a_1, b_2)\in A\times_M B$. By skew-symmetry of the bilinear form, $\beta(\Psi_1)(d)-\Psi_2(d)$ coincides with $-\beta(\Psi_2)(d') + \Psi_1(d')$, where now $d'$ should project to $(a_2, b_1)$, so we may take $d':=\sigma(d)$. Using~\eqref{eqn:sigma_from_beta}, we may cancel $\Psi_1(d')=\beta(\Psi_1)(d)$ and get $\Psi_2(d) = \beta(\Psi_2)(d')$, so $\Psi_2(d)=\beta(\Psi_2)(d') = \Psi_2(d'')$, where $d''=\sigma(d')=\sigma(\sigma(d))$. Since $\Psi_2$ is arbitrary, we get~\eqref{eqn:condition_on_sigma1}, this completes our proof.
\end{proof}

\begin{Remark} Note that the conditions \eqref{eqn:condition_on_sigma1} and \eqref{eqn:condition_on_sigma2} on $\sigma$ are symmetric with respect to the legs $D\rightarrow A$ and $D\rightarrow B$, whereas the skew-symmetric form on $E$ is not so. A question arise: what is a connection between $\sigma^\sv$ and $\sigma^\sh$? The later def\/ines a skew-symmetric form on $E^*:= \sD^{fH}= D^*_A$ which is the dual VB to $E$ over $C^*$. This suggests the answer: the skew-metrics on $E$ and $E^*$ def\/ined by $\sigma^\sv$ and $\sigma^\sh$, respectively, def\/ine the same identif\/ication $E\simeq E^*$:
\begin{gather*}
(\Phi, \Psi)\in \sigma^\sv \Leftrightarrow (\Psi, \Phi)\in \sigma^\sh.
\end{gather*}
Indeed, the former means that $\Phi(\sigma(d)) = \Psi(d)$ for any $d\in D$, while the latter is $\Psi(\sigma(d')) = \Phi(d')$ which is the same as $\sigma(\sigma(d)) = d$.
\end{Remark}

\begin{Remark}
Note that our observation provides a much simpler and shorter proof of the equivalence of the categories of metric DVBs and the category of $N$-manifolds of degree $2$ given in \cite{JotzLean:2015}, as well. All we need to do is to consider, instead of $\SymmVB[2]$, a similar category of DVBs $(D; A, B; M)$ equipped with a morphism $\sigma\colon \sD\rightarrow \sD^\f$ which satisf\/ies \eqref{eqn:condition_on_sigma1} and $\sigma|_{C} = - \id_C$ instead of~\eqref{eqn:condition_on_sigma2}. Then, the formula~\eqref{eqn:skew_form_on_E} changes to
\begin{gather}\label{eqn:symmetric_form_on_E}
\langle \Psi_1, \Psi_2\rangle_E = \dblInnerBracket{\sigma^{\sv}(\widetilde{\Psi_1}), \Psi_2} = \sigma^\sv(\widetilde{\Psi_1})(d)-\Psi_2(d) = - \Psi_1(\sigma(d)) - \Psi_2(d).
\end{gather}
Since $\sigma^\sv$ is minus the identity on $C^*$, we had to compose $\sigma^\sv$ with an DVB automorphism of $\sD^{\f\sv} = (D^*_B; B, C^*; M)$ which is minus the identity on the core. It is simply given by $\Psi \mapsto \widetilde{\Psi} = \underset{C^*}{-} \Psi$, $\Psi\in D^*_B$. It is straightforward to modify~\eqref{eqn:skew-symmetry} to see that the form given by~\eqref{eqn:symmetric_form_on_E} is symmetric and that DVB morphisms intertwining $\sigma$'s correspond to isotropic subbundles. A local form of~$\sigma$ satisfying~\eqref{eqn:condition_on_sigma1} and $\sigma|_{C}=-\id_C$ is $\sigma^\ast y^a_{01} = y^a_{10}$, $\sigma^\ast y^a_{10}= y^a_{01}$, $\sigma^\ast z^i = - z^i + y^a_{01} y^b_{10} \sigma^i_{ab }$ with \emph{symmetric} $\sigma^i_{ab}$: $\sigma^i_{ab}= \sigma^i_{ba}$. Thus we obtain an equivalence of categories of metric and skew-symmetric DVBs, \catname{skewSymmVB}[2], which dif\/fers from \SymmVB[2] only in that we impose $\sigma_C=-\id_C$ instead of~\eqref{eqn:condition_on_sigma1}. Then we can consider the linearisations of $N$-manifolds in a~similar manner as graded bundles, allowing the side bundles to be purely odd and prove that the linearisations of $N$-manifolds of degree $2$ are exactly skew-symmetric DVBs. We postpone details until Section~\ref{sec:Nmanifolds}.
\end{Remark}

\begin{Example}
Consider a vector bundle $\tau\colon E\rightarrow M$ equipped with a skew-symmetric form $\omega\colon E\times_{M} E \rightarrow \R$ and take its tangent lift
\begin{gather*}
\omega_{\sT E} := \dt \omega\colon \  \sT E\times_{\sT M} \sT E \rightarrow \R.
\end{gather*}
Then, $\sE = (\sT E; E, \sT M; M)$ is a symplectical DVB. What is the associated DVB morphism $\sigma$ which goes from $\sD = (\sT^*E; E, E^*; M)$ to its f\/lip $\sD^\f$? The form $\omega$ gives identif\/ications
 \begin{gather*}
 \tilde{\omega}\colon \ E\rightarrow E^*, \qquad \omega'\colon \ \sT^* E \rightarrow \sT^* E^*.
 \end{gather*}
 It is a matter of simple calculation to see that $\sigma$ is obtained by composing $\omega'$ with the canonical dif\/feomorphism $\tilde{\mathcal{R}}\colon \sT^*E^*\rightarrow \sT^* E$ (see~\cite{Dufour:1990, Mackenzie:1994, Tulczyjew:1977}) composed with an automorphism of $\sT^* E$ which is minus the identity on side bundle $E\rightarrow M$, and on the core bundle~$\sT^* M$, and so it is the identity on the side bundle $E^*\rightarrow M$, so that the composition $\sigma = \tilde{R}\circ \omega'$ is identity on the core bundle. In standard local coordinates $(x^a, y^i, p_a, p_i)$ on $\sT^* E$, $\sigma\colon \sT^*E \rightarrow \sT^* E$ reads as
 \begin{gather*}
 \sigma^*(x^a)=x^a, \qquad \sigma^*(y^j) = p_i  \omega^{ij}(x), \qquad \sigma^*(p_a) = p_a + y^k \omega_{kj}(x)\frac{\pa \omega^{ji}(x)}{\pa x^a} p_i,\\
  \omega^*(p_j) = y^i \omega_{ij}(x),
 \end{gather*}
 where $(\omega_{ij})$ is the matrix of the form $\omega$ and $(\omega^{ij})$ is its inverse. Due to skew-symmetry of $\omega$ the morphism $\sigma$ really squares to the identity.  We may similarly take a symmetric form on $E$ instead of $\omega$, def\/ine a metric DVB structure on $\sT E$, and f\/ind an analogues presentation of the associated morphism $\sigma$.
\end{Example}

\subsection{Lie algebroid structures}
A further observation here is that a symmetric DVB $(\sD, \sigma)$ gives rise to a Lie algebroid structure on the leg $D\rightarrow A$. To see this, take a graded bundle $F_2$ of degree $2$ and notice that the vertical subbundle $\sV F_2$ is a Lie sub-algebroid of the tangent Lie algebroid $\sT F_2$, and it induces a Lie algebroid structure on $\Linr(F_2) \rightarrow F_1$. Indeed, by very def\/inition of $\pLinr(F_2)$, a local section of the later can be naturally identif\/ied with a section $X$ of ${\sV F_2} \rightarrow F_2$ which is invariant with respect to the canonical action of the model vector bundle on the af\/f\/ine f\/ibration $\tau^2_1\colon F_2\rightarrow F_1$. In standard local coordinates $\big(x^A, y^a, z^i\big)$ on~$F_2$ it means that the vector f\/ield~$X$ has the form
\begin{gather*}
X=f^a(x, y)\pa_{y^a} + f^i(x, y) \pa_{z^i},
\end{gather*}
i.e., the coef\/f\/icient functions $f^a$, $f^i$ are constant on the f\/ibers of the f\/ibration $\tau^2_1$. We see immediately that such vector f\/ields are closed with respect to the Lie bracket of vector f\/ields. The anchor map $\rho\colon \pLinr(F_2)\rightarrow \sT F_1$ is simply $\sT \tau^2_1|_{\pLinr(F_2)}$, so $ \rho(X) = f^a(x, y)\pa_{y^a}$. Staring with an arbitrary symmetric DVB $(\sD, \sigma)$, we take as $F_2$ the diagonalisation of~$\sD$ def\/ined by~$\sigma$, and then transfer the obtained algebroid structure on the linearisation of $F_2$ to the vector bundle $D\rightarrow A$ by means of the isomorphism~\eqref{eqn:def_iso_I} $I\colon (\sD, \sigma) \rightarrow \Linr(F_2)\simeq \pLinr(F_2)$. A simple calculation shows that
\begin{gather*}
 [e_a, e_b] = \sigma_{ab}^i(x) f_i, \qquad [e_a, f_i]=0=[f_i, f_j], \qquad \rho(e_a) = \pa_{y_{01}^a}, \qquad \rho(f_i)=0,
\end{gather*}
where $(e_a, f_i)$ is a local basis of sections of $\pi^D_A\colon D\rightarrow A$ dual to $(y_{10}^a, z_{11}^i)$, and we use the local presentation of $\sigma$ as in~\eqref{eqn:sigma_ab_i}. It follows that the linear Poisson tensor on the dual bundle $D^*_A$, in standard f\/iber coordinates $(p_a^{10}, p_i^{11})$ dual to $(y_{10}^a, z_{11}^i)$, is given by
\begin{gather}\label{eqn:Poisson structure}
\Lambda = \frac12 \sigma_{ab}^i p_i^{11} \pa_{p_a^{10}} \wedge \pa_{p_b^{10}} + \pa_{p_a^{10}} \wedge \pa_{ y_{01}^a}.
\end{gather}
With respect to the standard gradation on the DVB $(D^*_A; A, C^*; M)$, the weights of $p_a^{10}$, $p_i^{11}$ are equal $(1,1)$ and $(1, 0)$, respectively. Therefore, the Poisson tensor $\Lambda$ is homogeneous of bi-weight~$(-1,-2)$.

Thus, we are lead to the following.
\begin{Proposition}
Any symmetric double vector bundle $(\sD, \sigma)$ is canonically equipped with the structure of a Lie algebroid on its leg $D \rightarrow A$ such that the associated Poisson structure on the dual double vector bundle~$D^*_A$ is homogeneous of bi-weight~$(-1,-2)$.
\end{Proposition}
By symmetry, we obtain a Lie algebroid on the leg $D\rightarrow B$, as well. However, both structure of Lie algebroid on the legs of $D$ do not constitute a structure of a double Lie algebroid on~$\sD$ since simply they are of bi-weights~$(-1, -2)$ and $(-2,-1)$, not~$(-1, -1)$. Note also, that we do not have the structure of a weighted Lie algebroid (and so not a~$\mathcal{VB}$-algebroid) as def\/ined in~\cite{Bruce:2014}.

\begin{Proposition}
The Poisson structure on~$D^{*}_{B}$ associated with a symmetric double vector bundle~$(\sD,\sigma)$ can be identified with the Poisson structure on the dual of the linearisation~$\pLinr^{*}(F_{2})$, which itself is the reduction of the canonical Poisson structure on~$\sT^{*}F_{2}$.
\end{Proposition}
\begin{proof}
The dual of the linearisation $\pLinr^{*}(F_{2})$ can be seen as the reduction of the cotangent bundle~$\sT^{*}F_{2}$ (cf.~\cite[Proposition~3.19]{Bruce:2014}) and the latter, of course, comes with a canonical symplectic structure. Let us employ local homogeneous coordinates on $\pLinr^{*}(F_{2})$ naturally inherited from canonical (Darboux) coordinates on~$\sT^{*}F_{2}$
\begin{gather*}\Big(\underbrace{x^{A}}_{(0,0)} ,\, \underbrace{y^{a,(1,0)}}_{(1,0)},\, \underbrace{\pi_{b}^{(0,1)}}_{(1,1)},\, \underbrace{\pi_{i}^{(1,1)}}_{(0,1)}\Big).\end{gather*}
In these canonical coordinates, the Poisson structure inherited from the reduction of the cano\-nical Poisson structure is
\begin{gather*}\Lambda_{\pLinr} = \partial_ {\pi_{a}^{(0,1)}} \wedge \partial_{ y^{a,(1,0)}},\end{gather*}
which by inspection is of bi-weight~$(-2,-1)$. One should also note that, although we have def\/ined the Poisson structure here in local coordinates, this def\/inition is invariant.

The isomorphism $I\colon \sD \rightarrow \pLinr(F_{2})$ gives rise to a dual isomorphism $\hat{I} \colon \pLinr^{*}(F_{2}) \rightarrow D^{*}_{B}$ which we write in local coordinates as
\begin{gather*}
\hat{I}^{*}y^{a}_{10} = y^{a,(1,0)}, \qquad \hat{I}^{*}p_{a}^{01} = \pi_{a}^{(0,1)} + \frac{1}{2!}y^{b,(1,0)}\sigma_{ba}^{i}(x)\pi^{(1,1)}_{i},\qquad \hat{I}^{*}p_{i}^{11} = \pi_{i}^{(1,1)}.
\end{gather*}
Thus we have
\begin{gather*}\hat{I}_{*} \Lambda_{\pLinr} = \frac12 \sigma_{ab}^i p_i^{11} \pa_{p_a^{01}} \wedge \pa_{p_b^{01}} + \pa_{p_a^{01}} \wedge \pa_{ y_{10}^a}, \end{gather*}
which, by comparison with~(\ref{eqn:Poisson structure}), is the desired Poisson structure.
\end{proof}

Because we are dealing with symmetric DVBs, application of the f\/lip functor means that we can also interpret the Poisson structure on $D^{*}_{A}$ as being inherited from the reduction of the canonical Poisson structure on~$\sT^{*}F_{2}$.

\subsection[The equivalence between graded bundles of degree $k$ and symmetric $k$-fold vector bundles]{The equivalence between graded bundles of degree $\boldsymbol{k}$\\ and symmetric $\boldsymbol{k}$-fold vector bundles}

Before giving careful statements and proofs let us sketch the basic idea. From~\cite{JG_MR_higher_vb} we know that any graded bundle $F_{k}$ can canonically be embedded in $\sT^{k}F_{k}$, and that $\sT^{k}F_{k}$ can be canonically embedded in~$\sT^{(k)}F_{k}$. \emph{Via} Proposition~\ref{prop:l_r_epi_and_inclusion} and Corollary~\ref{cor:L_emb_and_epim}, we see that the full linearisation of a~graded bundle $\Linr(F_{k})$ can be considered as a substructure of the iterated tangent bundle~$\sT^{(k)}F_{k}$. The latter has a number of canonical dif\/feomorphisms that can be enumerated by elements of the symmetric group $\Sgroup_{k}$. Thus the full linearisation of a graded bundle of degree $k$ is equipped with a number of canonical $k$-fold vector bundle morphisms $\sigma_{g}\colon \Linr(F_{k}) \rightarrow \Linr(F_{k})^{g}$, where $g \in \Sgroup_{k}$, and for any $k$-fold graded bundle $\sD =(D; \Delta_1, \ldots, \Delta_k)$ we def\/ine
\begin{gather*}\sD^g =\big(D; \Delta_{g(1)}, \ldots, \Delta_{g(k)}\big).
\end{gather*}
Note that $(\sD^{g_1})^{g_2} = \sD^{g_1g_2}$. A quick glimpse at the properties of $\sigma$'s lead us to the following def\/inition of a category~$\SymmVB[k]$, which will turn out to be equivalent with the category of graded bundles of degree $k$ (Theorem~\ref{thm:equivalence_k} below).

\begin{Definition}\label{def:SymmVB_k} Objects of the category $\SymmVB[k]$ are $k$-fold vector bundles $\sD = (D; \Delta_1, \ldots$, $\Delta_k)$ equipped with $k$-fold vector bundle isomorphisms $\sigma_g\colon \sD\rightarrow \sD^g$, where $g\in \Sgroup_k$, such that for each $g_1, g_2\in \Sgroup_k$ the composition
\begin{gather}\label{eqn:sigma1}
\sD \xrightarrow{\sigma_{g_2}} \sD^{g_2} \xrightarrow{(\sigma_{g_1})^{g_2}} \sD^{g_1g_2} = (\sD^{g_1})^{g_2}
\end{gather}
coincides with $\sigma_{g_1g_2}\colon \sD\rightarrow \sD^{g_1g_2}$ and
\begin{gather}\label{eqn:sigma2}
\sigma_{(i,j)}|_{C_{ij}} = \id_{C_{ij}},
\end{gather}
for any $1\leq i<j\leq k $, where $C_{ij}$ denotes the core bundle of the DVB $(D; \Delta_i, \Delta_j)$, and $(i,j)\in \Sgroup_k$ is the transposition swapping~$i$ and~$j$. The above $(\sigma_{g_1})^{g_2}$ coincides with $\sigma_{g_1}$ as a map $D\rightarrow D$ but we consider it here as a $k$-fold vector bundle morphism $\sD^{g_2} \rightarrow \sD^{g_1g_2}$. Note that the assignment $\sD \rightsquigarrow \sD^{g_2}$ is a functor between the categories of $k$-fold vector bundles, and it takes the mor\-phism~$\sigma_{g_1}$ to the morphism $(\sigma_{g_1})^{g_2}$. A morphism in $\SymmVB[k]$ from $(\sD, (\sigma_g)_{g\in \Sgroup_k})$ to $(\sD', ({\sigma'}_g)_{g\in \Sgroup_k})$ is a~$k$-fold vector bundle morphism $\phi\colon \sD\rightarrow \sD'$ such that
\begin{gather}\label{eqn:morphisms_in_C_k}
\sigma'_g \circ \phi = \phi^g \circ \sigma_g\colon \ \sD\rightarrow {\sD'}^g
\end{gather}
for any $g\in \Sgroup_k$.
\end{Definition}

\begin{Remark} One can give a def\/inition of symmetric $k$-fold vector bundle without setting the order of Euler vector f\/ields, and thus justifying the name ``symmetric". This is done by setting $T = \{t\colon \und{k}\to W\}$, where $\und{k}=\{1, \ldots, k\}$, $W=\{\Delta_1, \ldots, \Delta_k\}$ and considering the groupoid $\mathcal{G}\rightrightarrows T$ of all $k$-fold isomorphisms $\sD^t \to \sD^{t'}$, where $t, t'\in T$ and $\sD^t = (D; \Delta_{t(1)}, \ldots, \Delta_{t(k)})$. The symmetric group $\Sgroup_k$ acts canonically on $T$ and also on~$\mathcal{G}$. One can easily verify that giving a~groupoid morphism from the pair groupoid $T\times T\rightrightarrows T$ to $\mathcal{G}\rightrightarrows T$, covering $\id_T$, correspond with equipping $\sD$ with a family of DVB morphisms $\sigma_g$ as in Def\/inition~\ref{def:SymmVB_k} satisfying the condition~\eqref{eqn:sigma1}, but not necessary~\eqref{eqn:sigma2}.
\end{Remark}

\begin{Definition}
Let $(\sD, (\sigma_g)_{g\in \Sgroup_k})$ be a symmetric $k$-fold vector bundle, its \emph{diagonalisation} $F(\sD) \subset D$ is the graded bundle of degree $k$ constructed as the submanifold of $D$ consisting of all $\Sgroup_{k}$ invariant elements and the weight vector f\/ield is given by the total weight vector f\/ield $\Delta = \Delta_{1} + \cdots + \Delta_{k}$ restricted to the said submanifold.
\end{Definition}

It will be shortly seen that $\Delta$ is indeed tangent to $F(\sD)$ and the assignment $\sD \to F(\sD)$ is functorial. Moreover, the full linearisation functor and the diagonalisation functor are mutual inverses.

\begin{Example}
The pair $(\sT^{(k)}M, (\kappa_{g})_{g\in \Sgroup_k})$ is a symmetric $k$-fold vector bundle whose diagonalisation is clearly $\sT^{k}M$.
\end{Example}

\begin{Theorem} \label{thm:equivalence_k}
The full linearisation and the diagonalisation set up an equivalence of categories between $\GrB[k]$ and the category $\SymmVB[k]$.
\end{Theorem}
\begin{proof} Consider any graded coordinate system $\big(x^A, y^i_\veps\big)$ on $\sD$, where $1\leq i\leq N_\veps$, and the weight of $y^i_\veps$ is $\veps\in\{0,1\}^k$, $\veps\neq 0^k$. Since~$\sigma_g$ are dif\/feomorphisms of $D$ that mixes coordinates of the same total weight, the numbers $N_\veps$, $N_{\veps'}$ of coordinates of weights $\veps$ and $\veps'$ coincides if $|\veps|=|\veps'|$. Let $\Y{g}^i_\veps$ denote the same function on $D$ as $y^i_\veps$, but considered as a coordinate function on $\sD^g$, hence its weight is $\veps.g =(\veps_{g(1)}, \ldots, \veps_{g(k)})$ ($(\veps, g)\mapsto \veps.g$ is a right action of $\Sgroup_k$ on $\{0, 1\}^k$, and also on $\{0, 1\}^k \setminus \{0\}^k$). Def\/ine
\begin{gather*}
z^i_\veps := \frac{1}{k!}\sum_{g\in \Sgroup_k} \sigma^*_g \big(\Y{g}^i_{\veps.g^{-1}}\big).
\end{gather*}
Note that the weight of $\Y{g}^i_{\veps.g^{-1}}$ is $\veps.g^{-1}.g = \veps$, so $z^i_\veps$ has weight $\veps$ on $\sD$. Moreover, due to~\eqref{eqn:sigma2},
\begin{gather*}
\sigma_g^* \big(\Y{g}^i_\veps\big) = y^i_{\veps.g} + \text{lower terms}
\end{gather*}
holds for any transposition $g=(i, j)$, and so for any $g\in \Sgroup_k$, thanks to~\eqref{eqn:sigma1}, where \emph{lower terms} means the sum of products of at least two functions of non-zero weights that sum up to~$\veps.g$, hence they have to be of lower weight than $\veps.g$ with respect to the (partial) product order on~$\{0,1\}^k$.  It follows that
\begin{gather*}
z^i_\veps = y^i_\veps + \sum_{{j\geq 2, \,i_1, \ldots, i_j,}\atop {\veps^1+\cdots +\veps^j = \veps}} C^{i;\veps^1,\ldots,\veps^j}_{i_1, \ldots, i_j} y^{i_1}_{\veps^1} \cdots y^{i_j}_{\veps^j},
\end{gather*}
where the coef\/f\/icients $C^{i;\veps^1,\ldots,\veps^j}_{i_1, \ldots, i_j}$ are functions of $\big(x^A\big)$, so $\big(x^A, z^i_\veps\big)$ form a new graded coordinate system for $\sD$. Denote the corresponding coordinates on $\sD^g$ by $\Z{g}^i_\veps$. We shall show that
\begin{gather}\label{eqn:Zhi}
\sigma_h^* \big(\Z{h}^i_\veps\big) = z^i_{\veps.h}.
\end{gather}
Both sides of \eqref{eqn:Zhi} are of weight $\veps.h$, while abstracting from weights of coordinates we get
\begin{gather*}
\sigma_h^* (z^i_\veps) = \frac{1}{k!} \sum_{g\in \Sgroup_k} \sigma_h^* \sigma_g^* \big(y^i_{\veps.g^{-1}}\big) = \frac{1}{k!} \sum_{g\in \Sgroup_k} \sigma_{gh}^* \big(y^i_{\veps.g^{-1}}\big) = z^i_{\veps.h},
\end{gather*}
what justif\/ies our claim.

Thus the morphisms $\sigma_h\colon \sD\rightarrow \sD^h$ in new coordinates look extremely simple. The graded bundle $F_k\subset D$ associated with $(\sD, (\sigma_g)_{g\in \Sgroup_k})$ is given by
\begin{gather*}
F_k = \{q\in D\colon \sigma_g(q)=q \ \text{for any} \ g\in \Sgroup_k\} = \big\{\big(z^i_\veps\big)\colon z^i_\veps = z^i_{\veps'} \ \text{if}\ |\veps|=|\veps'|\big\},
\end{gather*}
and thus we may def\/ine without ambiguity graded coordinates on $F_k$ by $z^i_w := z^i_\veps|_{F_k}$, where $w=|\veps|$. Therefore, $F_k$ is indeed a graded bundle of degree $k$ and the diagonalisation is a~functor $\SymmVB[k] \to \GrB[k]$. Indeed, if $\psi\colon \sD\to \sD'$ is a morphism in~$\SymmVB[k]$, then for any $g\in\Sgroup_k$, and $q\in F_k$ holds $\sigma_g(\psi(q)) =\psi(\sigma_g(q)) = \psi(q)$ and~$\psi$ preserves the total weight given by $\Delta = \Delta_1+\cdots+\Delta_k$, hence $\psi|_{F_k}$ is a morphism in $\GrB[k]$ between diagonalisations of~$\sD$ and~$\sD'$.

Consider the full linearisation $\Linr{F_k}\subset \sT^{(k)} F_k$ with coordinates, denoted by $y^{i, (\veps)}$, inherited from $\sT^{(k)} F_k$. It is clear that $I\colon \sD\rightarrow \Linr(F_k)$, def\/ined locally by
\begin{gather*}
I^* y^{i, (\veps)} =z^i_\veps
\end{gather*}
is well def\/ined globally and makes the following diagram commutative:
\begin{gather*}
 \xymatrix{
 \sD \ar[rr]^{\sigma_g}\ar[d]_I && \sD^g \ar[d]^{I^g} \\
 \Linr F_k \ar[rr]^{\kappa_g|_{\Linr F_k}} && (\Linr F_k)^g. }
 \end{gather*}

We have described a natural one-to-one correspondence between objects in our categories. Now we shall point the correspondence of morphisms between these categories.

It is clear that a morphism of graded bundles $\phi\colon F_k\rightarrow F_k'$ gives rise to a morphism $\Linr{\phi}\colon$ $\Linr(F_k)\rightarrow \Linr(F_k')$ satisfying~\eqref{eqn:morphisms_in_C_k}. Conversely, assume $\psi\colon (\sD, (\sigma_g)) \rightarrow (\sD', (\sigma'_g))$ is a $k$-fold vector bundle morphism satisfying \eqref{eqn:morphisms_in_C_k}. It follows that we may restrict $\psi$ to $F_k\subset D$ and get a graded bundle morphism $\phi := \psi|_{F_k}\colon F_k\rightarrow F_k'$. Actually, the morphism $\psi\colon \sD\rightarrow \sD'$ is fully determined by its restriction to~$F_k$. To see this write a general form of~$\psi$ in the local coordinates $(z^i_\veps)$ def\/ined above:
\begin{gather*}
\psi^* z^{i'}_\veps = \sum_{{j\geq 1; \,k_1, \ldots, k_j,}\atop {\veps^1+\cdots +\veps^j = \veps}} Q^{i';\veps^1,\ldots,\veps^j}_{k_1, \ldots, k_j} z^{k_1}_{\veps^1} \cdots z^{k_j}_{\veps^j}, \qquad Q^{i';\veps^1,\ldots,\veps^j}_{k_1, \ldots, k_j}\in \R.
\end{gather*}
It follows from \eqref{eqn:morphisms_in_C_k} that the latter expression has to coincide with
\begin{gather*}
\sum_{{j\geq 1; \,k_1, \ldots, k_j,}\atop {\veps^1+\cdots +\veps^j = \veps}} Q^{i';\veps^1,\ldots,\veps^j}_{k_1, \ldots, k_j}   z^{k_1}_{\veps^1.g} \cdots z^{k_j}_{\veps^j.g}
\end{gather*}
for any $g\in \Sgroup_k$, so the coef\/f\/icients $Q^{\ldots}_{\ldots}$ must be symmetric in indices $k_a, k_b$ if $|\veps^a| = |\veps^b|$. Note that the formula for $\psi|_{F_k}$ involves only symmetric part of $Q^{\ldots}_{\ldots}$.
This completes our proof.
\end{proof}

Diagrammatically, we have
\begin{gather*}
\leavevmode
\begin{xy}
(0,20)*+{\catname{GrB}[k]}="a"; (50,20)*+{\catname{SymmVB}[k].}="b";%
{\ar@<1.ex>"a";"b"}?*!/_3mm/{\textnormal{full linearisation}};%
{\ar@<1.ex> "b";"a"} ?*!/_3mm/{\textnormal{diagonalisation}};%
\end{xy}
\end{gather*}

\begin{Example}
Following Example~\ref{exmpl:F3} of the full linearisation of a graded bundle of degree $3$, we see that the diagonalisation of $\Linr(F_{3})$, which is of course just $F_{3}$ itself, can be described in natural coordinates as
\begin{gather*}
 y^{a} := y^{a,(00)} = y^{a,(10)} = y^{(a),(01)}, \\ z^{i} := 2! z^{i,(10)}= 2! z^{i,(01)}= 2! z^{i,(11)}, \qquad w^{\kappa} := 3! w^{\kappa,(11)},
\end{gather*}
and by quick inspect it is clear that we recover the correct transformation laws for the non-zero weight coordinates. Up to combinatorial factors, one sets the coordinates of a given total weight equal: in a loose sense we look at the `diagonal submanifold' of the full linearisation. The generalisation of this local diagonalisation to higher degree graded bundles is clear.
\end{Example}

\section[$N$-manifolds and their linearisation]{$\boldsymbol{N}$-manifolds and their linearisation}\label{sec:Nmanifolds}
\subsection[The linearisation of $N$-manifolds]{The linearisation of $\boldsymbol{N}$-manifolds}
Given a supermanifold $\cM$, consider its tangent bundle $\sT \cM$. It is known that there exists a~ca\-no\-nical involution $\kappa_\cM$ of $\sT\sT \cM$, which in adapted coordinates looks the same as the classical (purely even) case. Let us denote local coordinates on $\cM$ by $x^A$, the \emph{parity} of $x^A$ by $ \widetilde{x}^{A} = \tilde{A}\in \Zet_2$. Then the adapted coordinates $\big(x^A, \dot{x}^A\big)$ on $\sT \cM$ transform in the usual way
\begin{gather*}
x^{A'} = x^{A'}(x), \qquad \dot{x}^{A'} = \dot{x}^B  \frac{\pa x^{A'}}{\pa x^B},
\end{gather*}
where $\dot{x}^A$ has the same Grassmann parity as $x^A$, i.e., $\widetilde{\dot{x}}^{A} = \widetilde{x}^{A} = \tilde{A}$. Thus we can employ on $\sT \sT \cM$ standard coordinates $\big(x^A, \dot{x}^A; \delta x^A, \delta \dot{x}^A\big)$ with additional transformation law of the form
\begin{gather}\label{eqn:TT_cM}
\delta {x}^{A'} = \delta {x}^B  \frac{\pa x^{A'}}{\pa x^B}, \qquad
\delta \dot{x}^{A'} = \delta{\dot{x}^B} \frac{\pa x^{A'}}{\pa x^B} + \delta x^C \dot{x}^B \,\frac{\pa^2 x^{A'}}{\pa x^B \pa x^C},
\end{gather}
hence the usual formula
\begin{gather*}
\kappa_\cM^*\big(\dot{x}^A\big) = \delta x^A, \qquad \kappa_\cM^*\big(\delta {x}^A\big) = \dot{x}^A, \qquad \kappa_\cM^*\big(\delta \dot{x}^A\big) = \delta \dot{x}^A
\end{gather*}
works f\/ine for supermanifolds.

Consider an $N$-manifold $\cF$ (of degree $k$ with base $M$) (cf.~\cite{Roytenberg:2002,Severa:2005}) on which we employ coordinates $\big(x^A, y^a_w\big)$ as usual, so the parity of $y^a_w$ is $w \bmod\, 2$ and $1\leq w\leq k$. Then, we take the vertical subbundle $\sV \cF$ of $\sT \cF$ which we consider as a (super) submanifold of $\sT \cF$ def\/ined locally by equations $\dot{x}^A=0$, but we change the weight as in the ordinary case. Thus $\dot{y}^a_w$ is of weight $w-1$ and the same parity as~$y^a_w$. Then, we may get rid of coordinates $y^a_k$ to obtain a local description of the \emph{linearisation}~$\pLinr(\cF)$ of $N$-manifold~$\cF$, exactly as in the case of (even) graded bundle. As possible bi-weights of our coordinates on $\pLinr(\cF)$ are $(w_1, w_2)$ with $w_1\in \{0,1\}$, $0 \leq w_2\leq k-1$, the linearisation $\pLinr(\cF)$ is a $\GrL$-bundle but in the category of supermanifolds. We may def\/ine the \emph{full linearisation} functor exactly as in even case as $\Linr(\cF) := \pLinr^{(k-1)}(\cF)$.

Consider the case $k=2$ and employ coordinates $\big(x^A, \zx^a, z^i\big)$ on $\cF$ of degrees $0$, $1$, $2$, respectively, with transformation law of the form
\begin{gather}\label{eqn:coord_law_F}
x^{A'} = x^{A'}(x), \qquad \zx^{a'} = \zx^{b}T_{b}^{\:\:\: a'}(x), \qquad z^{i'} = z^{j}T_{j}^{\:\:\: i'}(x) + \frac{1}{2}\zx^{a}\zx^{b}T_{ba}^{\:\:\:\: i'}(x),
\end{gather}
where $T_{ba}^{i'} = -   T_{ab}^{i'}$, as $\zx^a \zx^b = -  \zx^b \zx^a$. The formula
\begin{gather*}
i^*\big(x^A\big)=x^A, \qquad i^*(\dot{\zx}^a) = \zx^a, \qquad i^*(\delta \zx^a) = \dot{\zx}^a, \qquad i^*(\delta \dot{z}^i) = \dot{z}^i, \\
0=i^*\big(\dot{x}^A\big)=i^*(\delta x^A) = i^*\big(\delta\dot{x}^A\big) = i^*(\zx^a) = i^*(\delta \dot{\zx}^a) = i^*(z^i) = i^*(\dot{z}^i) = i^*(\delta z^i),
\end{gather*}
def\/ines a map $i\colon \pLinr(\cF)\rightarrow \sT\sT \cF$. Indeed, e.g., $i^*(\delta \dot{z}^{i'}) = \dot{z}^{i'} = \dot{z}^{j}T_{j}^{\:\:\: i'}(x) + \dot{\zx}^{a}\zx^{b}T_{ba}^{\:\:\:\: i'}(x)$, since using~\eqref{eqn:TT_cM} we f\/ind that $\delta \dot{z}^{i'} = \delta \dot{z}^j \frac{\pa z^{i'}}{\pa z^j} + \delta\zx^a\,\dot{\zx}^b \frac{\pa^2 z^{i'}}{\pa \zx^b \pa \zx^a}$ plus other terms on which $i^*$ vanishes, so~$i^*$ takes~$\delta \dot{z}^{i'}$ to $\dot{z}^{j}T_{j}^{\:\:\: i'}(x) + \dot{\zx}^{a}\zx^{b}T_{ba}^{\:\:\:\: i'}(x)$, the latter coincides with~$\dot{z}^{i'}$, by~\eqref{eqn:coord_law_F}. Thus we may obtain~$\pLinr(\cF)$ as a substructure of~$\sT \sT \cF$.  Furthermore, the canonical involution on the double tangent bundle of a supermanifold behaves exactly the same as in the classical case and hence~$\pLinr(\cF)$ is equipped with a canonical dif\/feomorphism $\sigma = {\kappa_{\cF}}|_{\pLinr(\cF)}$.

To describe $\sigma$ in terms of ordinary even dif\/ferential geometry, we shall consider the \emph{parity reversion functor} for double vector bundles (in the category of supermanifolds) as def\/ined in~\cite{Voronov:2012}.

\subsection{Parity reversion functor for double vector bundles}

All the necessary notions from the def\/inition of a double vector bundle carry over to the supermanifold setup. Thus a double vector bundle $(D; A, B; M)$ admits local homogeneous coordinates $\big(x^A; u^i, w^\alpha; z^\mu\big)$ of weights $(0,0)$, $(0,1)$, $(1,0)$, $(1,1)$, respectively, and of arbitrary parity. Within our notation, $(u^i)$, $(w^\alpha)$ are pullbacks of linear f\/iber coordinates of the side bundles $A\rightarrow M$, $B\rightarrow M$, respectively, and coordinate changes are of the form
\begin{gather}\label{eqn:superDtransformations}
x^{A'} = x^{A'}(x), \qquad u^{i'} = u^i T_{i}^{i'}, \qquad w^{\alpha'}= w^\alpha T_{\alpha}^{\alpha'}, \qquad z^{\mu'} = z^{\mu} T_{\mu}^{\mu'} + u^i w^\alpha T_{\alpha i}^{\mu'}.
\end{gather}
We apply the parity reversion functor to the vector bundle projection $D\rightarrow A$ in the f\/irst step and get a DVB denoted by $(\Pi_1 D; A, \Pi B; M)$. The index $1$ at $\Pi$ indicates that we take into account the f\/irst vector bundle structure on $D$ and $\Pi B \rightarrow M$ is the bundle obtained from the initial vector bundle $B\rightarrow M$ by changing parity to f\/iber coordinates. (This operation in not possible for graded bundles of higher degree.) Next we apply the parity reversion functor $\Pi_2$ to the second vector bundle structure on the obtained DVB and get a DVB $(\Pi_{21} D; \Pi A, \Pi B; M)$, where $\Pi_{21}=\Pi_2\circ \Pi_1$. Following explanation from~\cite{Voronov:2012}, we denote local coordinates on $(\Pi_1 D; A, \Pi B; M)$ and $(\Pi_2\Pi_1 D; \Pi A, \Pi B; M)$ by $\big(x^A; u^i, \eta^\alpha; \theta^\mu\big)$ and $\big(x^A; \zx^i, \eta^\alpha; z^\mu\big)$, respectively, and f\/ind that coordinate changes in $\Pi_{21}\sD$ have the form
\begin{gather}\label{eqn:pi2_pi1}
x^{A'} = x^{A'}(x), \qquad \zx^{i'} = \zx^i T_{i}^{i'}, \qquad \eta^{\alpha'}= \eta^\alpha T_{\alpha}^{\alpha'}, \qquad z^{\mu'} = z^{\mu} T_{\mu}^{\mu'} + (-1)^{\tilde{i}} \zx^i \eta^\alpha T_{\alpha i}^{\mu'},
\end{gather}
where $\tilde{i}$ is the parity of $u^i$, not $\zx^i$. Consider the square $S=\Pi_{21} \Pi_{21}$ of the functor $\Pi_{21}$. Although all the building bundles of $D$ and $S D$ can be canonically identif\/ied, it follows from~\eqref{eqn:pi2_pi1} that coordinates changes on $S \sD$ have the form
\begin{gather*}
x^{A'} = x^{A'}(x), \qquad u^{i'} = u^i T_{i}^{i'}, \qquad w^{\alpha'}= w^\alpha T_{\alpha}^{\alpha'}, \qquad z^{\mu'} = z^{\mu} T_{\mu}^{\mu'} - u^i w^\alpha T_{\alpha i}^{\mu'},
\end{gather*}
resulting from the fact that the parity of $\zx^i$ is $\tilde{i}+1\, \bmod\, 2$, and $(-1)^{\tilde{i}} (-1)^{\tilde{i}+1} = -1$. Although there is a canonical DVB isomorphism between $D$ and $SD$, namely given by
\begin{gather}\label{eqn:iso_D_sD}
I_\sD\colon \  \sD \rightarrow S\sD, \qquad\!\! I_{\sD}^*(x^A)=x^A, \qquad\!\! I_{\sD}^*(u^i)=u^i, \qquad\!\! I_{\sD}^*(w^\alpha) = w^\alpha, \qquad\!\! I_{\sD}^*(z^\mu)=-z^\mu,\!\!\!
\end{gather}
and although there exist many (non-canonical) DVB isomorphisms between $\sD$ and $S \sD$ inducing the identities on all building vector bundles of~$\sD$ (i.e., the core and on the side bundles), by referring to the theory of \emph{statomorphisms} introduced in \cite{Gracia-Saz:2009}, there is not, in general, a \emph{canonical} double vector bundle isomorphism between $\sD$ and $S\sD$ which induces the identities on all building vector bundles. This is in full analogy to DVB isomorphisms between the cotangent bundles~$\sT^* E$ and~$\sT^* E^*$ of total spaces of a vector bundle $E\rightarrow M$ and its dual~$E^*$.

\begin{Remark}[\cite{Voronov:2012}] The functors $\Pi_{12}$ and $\Pi_{21}$ are dif\/ferent in the sense presented in \cite{Gracia-Saz:2009}, as they are linked by the functor $S$.
\end{Remark}

Another observation is that, in general, there is no canonical statomorphism between $\Pi_{21} \sD^\f$ and $(\Pi_{21} \sD)^\f$, but the following holds.

\begin{Proposition}
There is a canonical statomorphism between $\Pi_{21} (\sD^\f)$ and $S (\Pi_{21} \sD)^\f$.
\end{Proposition}
\begin{proof} Let us write the transformation law \eqref{eqn:superDtransformations} for coordinates $(z^\mu)$ on $\sD^\f$ in the form
\begin{gather*}
z^{\mu'} = z^{\mu} T_{\mu}^{\mu'} + w^\alpha u^i \hat{T}_{i \alpha}^{\mu'},
\end{gather*}
where $\hat{T}_{i \alpha}^{\mu'} = (-1)^{\tilde{i}\tilde{\alpha}} T_{\alpha i}^\mu$. Then, according to \eqref{eqn:pi2_pi1}, for coordinates $(z^\mu)$ on $\Pi_{21} \sD^\f$ we have
\begin{gather*}
z^{\mu'} = z^{\mu} T_{\mu}^{\mu'} + (-1)^\alpha \eta^\alpha \zx^i \hat{T}_{i \alpha}^{\mu'}.
\end{gather*}
Directly from \eqref{eqn:pi2_pi1} we f\/ind that coordinates $(z^\mu)$ on $(\Pi_{21} \sD)^\f$ transform as
\begin{gather*}
z^{\mu'} = z^{\mu} T_{\mu}^{\mu'} + \epsilon   \zx^i \eta^\alpha T_{\alpha i}^{\mu'},
\end{gather*}
where $\epsilon = (-1)^{\tilde{i}} (-1)^{(\tilde{i}+1)(\tilde{\alpha}+1)}$. Hence, the functors $\Pi_{21}\circ \f$ and $\f\circ \Pi_{21}$ dif\/fers by $S$, as $(-1)^{\tilde{\alpha}} (-1)^{\tilde{i}\tilde{\alpha}} = - \epsilon$, and the result follows.
\end{proof}

\subsection[$N$-manifolds of degree 2 and metric double vector bundles]{$\boldsymbol{N}$-manifolds of degree 2 and metric double vector bundles}

Apply the parity reversion functor $\Pi_{21}$ to the linearisation of an $N$-manifold $\cF$ of degree $2$ and the morphism $\sigma\colon \Linr(\cF) \rightarrow \Linr(\cF)^\f$ def\/ined by $\kappa_\cF\colon \sT\sT \cF\to \sT\sT \cF$ and the inclusion $\Linr(\cF)\subset \sT\sT \cF$. We shall obtain an ordinary (even) double vector bundle $\sD = (D; A, B; M) = \Pi_{21} \Linr(\cF)$ with isomorphic side bundles $A\simeq B\simeq \Pi \cF_1$ and a DVB morphism $\Pi_{21}(\sigma)\colon \sD \rightarrow \Pi_{21} (\Linr(\cF)^\f)$ which is the identity on all the building bundles as~$\kappa_\cF$ is so. To get a morphism $\sD\rightarrow \sD^\f$ we compose~$\Pi_{21}(\sigma)$ with a DVB isomorphism \eqref{eqn:iso_D_sD}
\begin{gather*}
\Pi_{21} \big(\Linr(\cF)^\f\big) \simeq S (\Pi_{21} \Linr(\cF))^\f \xrightarrow{\eqref{eqn:iso_D_sD}} (\Pi_{21} \Linr(\cF))^\f = \sD^\f.
\end{gather*}
Due to properties of above isomorphisms, the obtained DVB morphism $\hat{\sigma}\colon \sD \rightarrow \sD^\f$ is minus the identity on the core and the identity on the side bundles. A morphism $\phi\colon \cF \rightarrow \cF'$ between $N$-manifolds $\cF$ and $\cF'$ of degree $2$ is given locally by
\begin{gather*}
x^{A'} = x^{A'}(x), \qquad \zx^{a'} = \zx^{b}Q_{b}^{\:\:\: a'}(x), \qquad z^{i'} = z^{j}Q_{j}^{\:\:\: i'}(x) + \frac{1}{2}\zx^{a}\zx^{b}Q_{ba}^{\:\:\:\: i'}(x),
\end{gather*}
where $Q_{ba}^{\:\:\:\: i'} = - Q_{ab}^{\:\:\:\: i'}$, and it results in a DVB morphism $(\Pi_{21}\Linr)(\phi): \sD \rightarrow \sD'$ of the form
\begin{gather*}
x^{A'} = x^{A'}(x), \qquad y^{a'}_{01} = y^{b}_{01} Q_{b}^{\:\:\: a'}(x), \qquad y^{a'}_{10} = y^{b}_{10} Q_{b}^{\:\:\: a'}(x), \\
z^{i'} = z^{j}Q_{j}^{\:\:\: i'}(x) - y^b_{01} y^a_{10} Q_{ab}^{\:\:\:\: i'}(x),
\end{gather*}
which obviously `respects' additional structures $\hat{\sigma}$, $\hat{\sigma'}$.
This lead us to introducing a category $\catname{skewSymmVB}[2]$ of \emph{skew-symmetric DVBs}, whose objects are ordinary double vector bundles $\sD$ equipped with a DVB morphism $\sigma\colon \sD \rightarrow \sD$ with only such dif\/ference to the category $\SymmVB[2]$ that~\eqref{eqn:condition_on_sigma2} is replaced with $\sigma_C= - \id_C$ where $C$ is the core of $\sD$. As we have mentioned earlier in the paper, the category $\catname{skewSymmVB}[2]$ just introduced is equivalent to the category of metric DVBs (cf.~\cite{JotzLean:2015}). Thus we have proved
\begin{Proposition}\label{prop:Jotzlean}
The category of metric double vector bundles is equivalent to the category $\catname{skewSymm}[2]$, and hence is also equivalent to the category of $N$-manifolds of degree~$2$.
\end{Proposition}

\section{Superisation of graded bundles}\label{sec:superisation}

 \subsection{Using the parity shifted full linearisation}
 It is clear that we can use the total linearisation and the parity reversion functor in each `direction' to associate with a graded bundle an $N$-manifold (cf.~\cite{Grabowski:2006,Voronov:2012}). Of course this is not a categorical equivalence in general, but we do have a categorical equivalence when we restrict attention to the subcategory of $N$-manifolds that arise as the superisation of symmetric $k$-fold vector bundles.

 \begin{Proposition}
 Let $F_{k}$ be a graded bundle of degree $k$, $\Pi^{(k)}\Linr(F_{k})$ is a $N$-manifold of degree~$k$, where $\Pi^{(k)} := \Pi_{k} \Pi_{k-1} \cdots \Pi_{1}$ is the total parity reversion functor. Moreover, this association is functorial.
 \end{Proposition}

However, one has to be carefully applying parity reversion functors to $k$-fold vector bundles. A choice in ordering the coordinates as they appear in the transformation law or, equivalently, a choice which order to apply the parity reversion functors is needed. Dif\/ferent choices here will result in spurious minus sign dif\/ferences at the level of coordinates and their transformation rules. We will obtain dif\/ferent functors for each of these choices, however they will be related via natural transformations. From our perspective, the fact that one has to make choices and consider natural transformations here is somewhat unpleasing. For a full discussion con\-sult~\cite{Grabowski:2006,Voronov:2012}.

Note that if one restricts attention to vector bundles, then the subtleties explained above do not appear. The standard notion of a superisation works \emph{perfectly} for vector bundles, and this suggests that some more general notion of a `superisation' is needed for $k$-fold vector bundles.

\subsection[Using $\mathbb{Z}_{2}^{k}$-supermanifolds]{Using $\boldsymbol{\mathbb{Z}_{2}^{k}}$-supermanifolds}
As an alternative to the standard parity reversion functor, one can use the full linearisation to def\/ine a $\mathbb{Z}_{2}^{k}$-supermanifold by declaring the commutation rules of the coordinates to be def\/ined by the scalar product of the weights. The reader should consult \cite{Covolo:2014a,Covolo:2014b} for the locally ringed space approach, noting that Molotkov \cite{Molotkov:2010} developed his notions in the categorical framework of the functor of points. We denote the category of $\mathbb{Z}_{2}^{k}$-supermanifolds with $\catname{SMan}[k]$.

From \cite[Proposition~5.5]{Covolo:2014a} we know that the superisation of a $k$-fold vector bundle, $(k >1)$ is a $ \mathbb{Z}_{2}^k$-supermanifold.
More specif\/ically, in local coordinates $(y^{a, (\ep)}_w)_{w=|\ep|}$ on $\Linr(F_{k})$ one can def\/ine a $\mathbb{Z}_{2}^{k}$-supermanifold using the sign rule
\begin{gather*}
Y^{a, (\ep)}_w Y^{b, (\delta)}_v = (-1)^{\langle \ep , \delta \rangle}Y^{b, (\delta)}_v Y^{a, (\ep)}_w,
\end{gather*}
where $\langle \ep , \delta \rangle=\sum\ze_i\zd_i$ is the standard scalar product on $\mathbb{Z}_{2}^{k}$: to each $y$ one assigns a correspon\-ding~$Y$ of the same multi-weight. The \emph{Grassmann parity} of a given coordinate $Y$ is def\/ined as the total weight of the coordinate, but note that the parity does \emph{not} determine the commutation rules.

Let us denote the $\mathbb{Z}_{2}^{k}$-supermanifold build from the full linearisation of a graded bundle with~$\Pi \Linr(F_{k})$. As the changes of local coordinates on $\Linr(F_{k})$ respect the weight, this superisation is well-def\/ined. Furthermore, the changes of local coordinates on $\Pi \Linr(F_{k})$ are exactly the same as for the full linearisation itself. As we have assumed $F_{k}$ to be a manifold, there are no potential minus signs that arise via any reordering, since the changes of coordinates never involve coordinates that do not strictly commute. Evidently, the graded commutation rules for the local coordinates do not ef\/fect the form of the transformation laws and so the local gluings of coordinate charts. There are no choices by hand here that manifest themselves as specious minus signs; we can all agree on $\Pi \Linr(F_{k})$ without any further choices or conventions.

If we start with a vector bundle in the category of smooth manifolds, say $E$, then it is easy to see that the full linearisation acts as the identity functor, i.e., $\Linr(E) \simeq E$. Thus the $\mathbb{Z}_{2}^{k}$-superisation reduces to the standard $\mathbb{Z}_{2}$-superisation via the parity reversion functor.

We summarise the main observations and results of this paper in the following f\/inal theorem.

\begin{Theorem}\label{thm:superisation}
There exists a canonical superisation functor
\begin{gather*}\Pi \Linr\colon \ \catname{GrB}[k] \longrightarrow \catname{SMan}[k], \end{gather*}
defined as the composition of the full linearisation functor and the canonical $\mathbb{Z}_2^k$-superisation functor of $k$-fold vector bundles.
\end{Theorem}

\subsection*{Acknowledgements}
The authors thank the anonymous referees whose comments and suggestions have served to improve the presentation of this work. Research funded by the Polish National Science Centre grant under the contract number DEC-2012/06/A/ST1/00256.

\LastPageEnding

\end{document}